\documentclass{llncs}

\usepackage{amssymb,amsmath}
\usepackage{graphicx,color}
\usepackage{comment}
\newcommand{\etal}{et al.~}

\DeclareMathOperator{\operatorClassNP}{NP}
\newcommand{\classNP}{\ensuremath{\operatorClassNP}}

\DeclareMathOperator{\operatorClassW}{W}
\newcommand{\classW}[1]{\ensuremath{\operatorClassW[#1]}}

\begin{document}

\title{Induced Disjoint  Paths in AT-free Graphs\thanks{This work is supported by EPSRC (EP/G043434/1), Royal Society (JP100692), the European Research Council (ERC) via the grant Rigorous Theory of Preprocessing (reference 267959),  and ERC StG project PAAl no.\ 259515. A preliminary abstract of this paper appeared in the proceedings of SWAT 2012~\cite{GPV12-swat}.}}
\author{Petr A. Golovach\inst{1} 
\and 
Dani\"el Paulusma\inst{2}
\and
Erik Jan van Leeuwen\inst{3}
}
\institute{
Department of Informatics, University of Bergen, PB 7803, 5020 Bergen, Norway\\
\texttt{petr.golovach@ii.uib.no}.
\and
School of Engineering and  Computing Sciences, Durham University\\
Science Laboratories, South Road, Durham DH1 3LE, UK\\
\texttt{daniel.paulusma@durham.ac.uk}.
\and
Department of Information and Computing Sciences, Utrecht University\\
PO Box 80.089,
3508 TB Utrecht,
The Netherlands\\
\texttt{e.j.van.leeuwen@uu.nl}.
}
\maketitle

\begin{abstract}
Paths $P_1,\ldots,P_k$ in a graph~$G=(V,E)$ are mutually induced if 
any two distinct $P_i$ and $P_j$  have neither common vertices nor adjacent vertices (except perhaps their end-vertices). 
The {\sc Induced Disjoint Paths} problem is to decide if
a graph~$G$ with $k$ pairs of specified vertices $(s_i,t_i)$ contains 
$k$ mutually induced paths $P_i$ such that each $P_i$ connects $s_i$ and $t_i$.
This is a classical graph problem that is \classNP-complete even for $k=2$. We study it for AT-free graphs.

\medskip
\noindent
Unlike its subclasses of permutation graphs and cocomparability graphs, the class of AT-free graphs has  
no known characterisation by a geometric intersection model.
However, by a new, structural analysis of the behaviour of  {\sc Induced Disjoint Paths} for AT-free graphs,
we prove that it can be solved in polynomial time for AT-free graphs even when $k$ is part of the input.
This is in contrast to the situation for other well-known graph classes, such as planar graphs, claw-free graphs, or more recently, (theta,wheel)-free graphs (JCTB 2021), for which such a result only holds if $k$ is fixed. 

\medskip
\noindent
As a consequence of our main result, the problem of deciding if a given AT-free graph contains a fixed graph~$H$ as an induced
 topological minor admits a polynomial-time algorithm. In addition, we show that such an algorithm  
 is essentially optimal by proving that the problem is  \classW{1}-hard with parameter~$|V_H|$, even on a subclass of AT-free graphs, namely cobipartite graphs.
We also show that the problems {\sc $k$-in-a-Path} and {\sc $k$-in-a-Tree} are polynomial-time solvable on AT-free graphs even if $k$ is part of the input.
These problems are to test if a graph has an induced path or induced tree, respectively, spanning $k$ given vertices.
\end{abstract}

\section{Introduction}\label{s-intro}

We study the {\sc Induced Disjoint Paths} problem. This problem can be seen as the induced version of the well-known {\sc Disjoint Paths} problem, which is to decide if a graph~$G$ with $k$ pairs of specified vertices $(s_i,t_i)$ contains a set of $k$ mutually vertex-disjoint paths $P_1,\ldots,P_k$, where each $P_i$ starts from $s_i$ and ends at $t_i$.
The {\sc Disjoint Paths} problem is one of the problems in Karp's list of \classNP-compete problems~\cite{Ka75}.
If $k$ is a {\it fixed} integer, that is, not part of the input, then the problem can be solved in $O(n^3)$ time for $n$-vertex graphs, as shown by Robertson and Seymour~\cite{RS95}
(see~\cite{KKR12} for an algorithm with $O(n^2)$ running time).
In contrast, the {\sc Induced Disjoint Paths} problem is \classNP-complete even if $k=2$, as shown both by Fellows~\cite{Fe89} and  Bienstock~\cite{Bi91}.

The {\sc Induced Disjoint Paths} problem is formally defined as follows. The problem takes as input a graph~$G$ and a set~$S=\{(s_{1},t_{1}),\ldots,(s_{k},t_{k})\}$ of $k$ distinct unordered pairs 
of specified vertices of a graph~$G$, which we call {\it terminals}. It asks if $G$ contains $k$ 
paths $P_{1},\ldots,P_{k}$ that satisfy the following conditions:
\begin{itemize}
\item [1.] for every $i$, $P_i$ has end-vertices $s_i$ and $t_i$;
\item [2.] if $i\neq j$, then $P_i$ and $P_j$ may 
share at most one vertex, which must be an end-vertex of both $P_i$ and $P_j$;
\item [3.] if $i\neq j$, then no inner vertex~$u$ of $P_i$ is adjacent to a vertex~$v$ of $P_j$
except when $v$ is an end-vertex of both $P_i$ and $P_j$;
\item [4.] for every $i$, $P_i$ is an induced path.
\end{itemize}
A collection of paths $P_1,\ldots,P_k$ that satisfies the aforementioned 
conditions 1--4 for a given graph~$G$ with a set $S$ of $k$ terminals pairs $(s_i,t_i)$ is called a {\it solution} for $(G,S)$. An example of a graph with three paths that satisfy all conditions is shown in Fig.~\ref{fig:dp}.

\begin{figure}[ht]
\centering\scalebox{0.7}{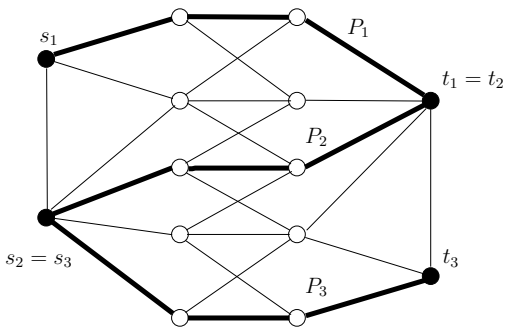}
\caption{An example of an instance $(G,S)$ of {\sc Induced Disjoint Paths}, which is a yes-instance as the three paths $P_1,P_2,P_3$ satisfy conditions 1--3. 
\label{fig:dp}}
\end{figure}

\noindent
We make two observations about the problem definition. 

\medskip
\noindent
{\bf Remark 1.} 
We require the terminal pairs to be distinct, but do allow in conditions~2 and~3 that  
two terminals from distinct pairs coincide or are adjacent.  Hence, conditions~2 and~3 are slightly more relaxed compared to the usual definition of {\sc Induced Disjoint Paths} by means of mutually induced paths where paths in a solution have neither common nor adjacent vertices. For general graphs, we can make these relaxations without loss of generality, because we can perform some elementary graph operations to ensure that two terminals from different pairs are distinct and non-adjacent. However, these operations may break membership of a graph class. 
Hence, our definition with the relaxed conditions yields a more general problem to solve if we restrict the input to some special graph class. We need this more general problem for our applications, as we will explain in 
Section~\ref{s-appp}.\footnote{We refer to Open Problem~1 in Section~\ref{s-con} for a discussion on the even more general problem variant where two terminal pairs are allowed to be the same.}

\medskip
\noindent
{\bf Remark 2.}
As we can take shortcuts if necessary, we may omit condition~4 to obtain an equivalent problem. However, we added condition~4 explicitly, as this will be convenient for algorithmic purposes. 

\medskip
\noindent
{\bf Research Question.} 
As the {\sc Induced Disjoint Paths} problem is \classNP-complete even if $k=2$~\cite{Bi91,Fe89}, we study the following research question to increase our understanding in the computational complexity of the problem:

\medskip
\noindent
{\it What are natural graph classes for which {\sc Induced Disjoint Paths} is polynomial-time solvable?}

\subsection{Known Results}
For planar graphs, {\sc Induced Disjoint Paths} stays \classNP-complete. This result can be obtained by subdividing each edge of a planar input graph of {\sc Disjoint Paths}. As planar graphs are closed under edge subdivision, this yields a planar input 
graph of {\sc Induced Disjoint Paths}. We then apply the result of Lynch~\cite{Ly75} who proved that {\sc Disjoint Paths} is \classNP-complete  for planar graphs.
Kobayashi and Kawarabayashi~\cite{KK12} 
presented an algorithm that solves {\sc Induced Disjoint Paths} on planar graphs that runs in linear time for any fixed $k$; see~\cite{KK09} for an extension of this result to graph classes of bounded genus.

For claw-free graphs, {\sc Induced Disjoint Paths} stays \classNP-complete as well. This is shown by Fiala et al.~\cite{FKLP12} even for line graphs, which form a subclass of claw-free graphs, 
and very recently by Radovanovi\'c, Trotignon and Vu\v{s}kovi\'c~\cite{RTV19} even for line graphs of triangle-free chordless graphs.
In addition, Fiala et al.~\cite{FKLP12} showed that  {\sc Induced Disjoint Paths} can be solved in polynomial time for claw-free graphs if $k$ is fixed.  
In a previous paper~\cite{GPV12-esa},  
we improved the latter result by showing that {\sc Induced Disjoint Paths} is fixed-parameter tractable for claw-free graphs when parameterized by $k$.
Radovanovi\'c, Trotignon and Vu\v{s}kovi\'c~\cite{RTV19} proved that {\sc Induced Disjoint Paths} can be solved in polynomial time for
(theta,wheel)-free graphs if $k$ is fixed. However, the problem stays \classNP-complete for (theta,wheel)-free graphs if $k$ is part of the input due to the aforementioned \classNP-completeness result for the subclass of line graphs of triangle-free chordless graphs~\cite{RTV19}.

From the above we observe that being planar, claw-free or (theta,wheel)-free does not help with respect to polynomial-time solvability of {\sc Induced Disjoint Paths} when $k$ is part of the input. However, for chordal graphs, {\sc Induced Disjoint Paths} is polynomial-time solvable when $k$ is part of the input, as shown by Belmonte et al.~\cite{BGHHKP12}, whereas we showed that for circular-arc graphs, {\sc Induced Disjoint Paths} is even linear-time solvable~\cite{GPV16}.

More recently, Jaffke, Kwon and Telle~\cite{JKT17} proved that {\sc Induced Disjoint Paths} is polynomial-time solvable for any graph class of bounded mim-width. 
Their results imply that in addition to the above graph classes, the problem is polynomial-time solvable for the following classes:  bi-interval graphs, circular permutation graphs, convex graphs, 
circular $p$-trapezoid graphs, $p$-polygon graphs, Dilworth-$p$ graphs, and co-$p$-degenerate graphs. 

\subsection{Our Main Result}
In this paper, we prove that {\sc Induced Disjoint Paths} is polynomial-time solvable on
a large class of graphs, namely
the class of {\it asteroidal triple-free graphs} (or {\it AT-free graphs}), even if $k$ is part of the input. An {\it asteroidal triple} is a set of three mutually non-adjacent vertices, such that each two of them are joined by a path that avoids the neighbourhood of the third; AT-free graphs are exactly those graphs that contain no such triple~\cite{LB62}.
The class of AT-free graphs is easily recognized in polynomial time, widely studied (see e.g.~\cite{KKM01,Kr00,KM12,KMT03,St10}), and contains many well-known classes, such as cobipartite graphs, cocomparability graphs, cographs, interval graphs, permutation graphs, and $p$-trapezoid graphs (cf.~\cite{CorneilOS97}). 
In contrast to these subclasses, AT-free graphs miss the advantage of having a known geometric intersection model (which can be exploited for the design of polynomial-time algorithms).
The class of AT-free graphs is incomparable to chordal graphs, circular-arc graphs,
and graphs of bounded mim-width,
which were covered by previous work, as mentioned above. Hence, with our main result, we make significant progress on answering the research question.

It is interesting to observe here that {\sc Induced Disjoint Paths} seems to behave substantially different from {\sc Disjoint Paths}. Generally, the induced variant of a containment relation problem is computationally just as hard or even harder than its non-induced variant. 
For example, {\sc Matching} can be solved in polynomial time, but {\sc Induced Matching} is \classNP-hard 
even for bipartite graphs~\cite{Ca89}.
In contrast,
{\sc Disjoint Paths} is \classNP-complete on split graphs~\cite{HHSV15}, but {\sc Induced Disjoint Paths} is polynomial-time solvable on the strictly larger class of chordal graphs~\cite{BGHHKP12}. Similarly, {\sc Disjoint Paths} is \classNP-complete on interval graphs~\cite{NatarajanS96}, but {\sc Induced Disjoint Paths} is polynomial-time solvable on the strictly larger classes of circular-arc graphs~\cite{GPV16} and chordal graphs~\cite{BGHHKP12}. With our main result, we now add the strictly and 
substantially
larger class of AT-free graphs to the list of classes for which {\sc Induced Disjoint Paths} is polynomial-time solvable, but {\sc Disjoint Paths} is not~\cite{NatarajanS96}, unless P$=$\classNP.

In order to solve {\sc Induced Disjoint Paths} in polynomial time on AT-free graphs, we first apply some general preprocessing rules (Section~\ref{sec:preproc}). We then provide a thorough exploration of the structure of AT-free graphs, in particular in relation to the {\sc Induced Disjoint Paths} problem (Section~\ref{sec:struct}). Finally, we use these structural results for a 
dynamic programming algorithm (Section~\ref{sec:dp}). Here we heavily rely on the fact that AT-free graphs have a path that dominates the graph and that AT-free graphs cannot contain long induced cycles, so that we can always focus on the neighbourhood of a constant number of vertices in the graph.

Our approach is substantially different from the approach of Belmonte~\etal\cite{BGHHKP12} for solving this problem in polynomial time on chordal graphs, as it appears that the tree-decomposition-based approach of~\cite{BGHHKP12} does not work for AT-free graphs. 
Our proof techniques are also quite different from the irrelevant vertex technique used by 
Kobayashi and Kawarabayashi~\cite{KK12} for planar graphs, which only work for fixed $k$ 
in contrast to our result, 
and from the proof techniques for claw-free graphs~\cite{FKLP12,GPV12-esa}.
The latter techniques are based on the characterization of claw-free graphs of Chudnovsky and Seymour~\cite{CS05} and also only work when $k$ is 
a fixed constant instead of being part of the input.

\subsection{Applications to Other Containment Relation Problems}\label{s-appp}
We show that our algorithm for {\sc Induced Disjoint Paths} for AT-free graphs can be used as a subroutine to solve two categories of containment relation problems in polynomial-time for AT-free graphs. Recall that a containment relation problem asks to detect if one graph is contained in some other graph by some specified graph containment relation. For general graphs, {\sc Induced Disjoint Paths} has no direct applications for other containment problems, as it is \classNP-complete when $k=2$~\cite{Bi91,Fe89}. However, such applications are possible for graph classes
for which {\sc Induced Disjoint Paths} is polynomial-time solvable, such 
as AT-free graphs.

Below we provide a definition, survey previous work, and state our results for both considered categories of containment relation problems. Proof details can be found in Sections~\ref{s-itm} and~\ref{s-path+tree}.

\medskip
\noindent
{\bf \boldmath$H$-Induced Topological Minor.}
A graph~$G$ contains a graph~$H$ as a {\it topological (induced) minor} if 
$G$ contains an (induced) subgraph that is isomorphic to a subdivision of $H$, that is, to a graph obtained from $H$ by a number of edge subdivisions.
The problems that are to decide whether a given graph contains some fixed graph~$H$ as a topological minor or induced topological minor are 
called {\sc $H$-Topological Minor} and {\sc $H$-Induced Topological Minor}, respectively. 

There is a great disparity in what we know about {\sc $H$-Topological Minor} and {\sc $H$-Induced Topological Minor}. Robertson and Seymour~\cite{RS95} showed that {\sc $H$-Topological Minor} can be solved in polynomial time for any fixed graph~$H$. 
Grohe et al.~\cite{GKMW11} improved this result to cubic time. In contrast, for {\sc $H$-Induced Topological Minor}, L\'{e}v\^{e}que et al.~\cite{LLMT09} obtained some partial results by showing both polynomial-time and \classNP-complete cases. 
In particular they proved that  {\sc $H$-Induced Topological Minor} is \classNP-complete for $H=K_5$.
Afterwards, Le~\cite{Le17} proved that {\sc $H$-Induced Topological Minor} is polynomial-time solvable for $H=K_4$, which was known to be a stubborn case (see also~\cite{LMT12}). However, the complexity classification for {\sc $H$-Induced Topological Minor} is still 
far from complete.

In contrast to the situation for general graphs, we can determine the complexity 
of {\sc $H$-Induced Topological Minor} for any graph~$H$ when we restrict the input to AT-free graphs. 
For doing this, we make explicit use of conditions 2 and~3 for the terminals in the definition of {\sc Induced Disjoint Paths}. 
We first consider the anchored version of $H$-{\sc Induced Topological Minor}. 
This variant is to decide if a  graph~$G$ contains a graph~$H$ as an induced topological minor such that the vertices of $H$ are mapped to specified vertices of $G$. 
We show that the anchored version of $H$-{\sc Induced Topological Minor} can be solved in polynomial time on AT-free graphs even when $H$ is an arbitrary graph that is part of the input.
This result cannot be generalized to the original, non-anchored version of this problem, 
which we show to be \classNP-complete even for a subclass of AT-free graphs, namely for the class of cobipartite graphs. 

However, our result for the anchored version still implies a polynomial-time algorithm for {\sc $H$-Induced Topological Minor} on AT-free graphs as long as $H$ is fixed. This may be the best result one can hope for: besides proving \classNP-completeness, we also prove in Section~\ref{s-itm} that {\sc Induced Topological Minor} is \classW{1}-hard when parameterized by $|V_H|$ even for cobipartite graphs.

We note that by demanding conditions~2 and~3 and following the above approach based on the anchored version of the problem, $H$-{\sc Induced Topological Minor} (for fixed $H$) is polynomial-time solvable for any 
graph class for which {\sc Induced Disjoint Paths} is polynomial-time solvable or becomes polynomial-time solvable after fixing~$k$; see, for example,~\cite{BGHHKP12,FKLP12,JKT17,RTV19} where this has been shown for chordal graphs, claw-free graphs, graphs of bounded mim-width and (theta,wheel)-free graphs, respectively.

\medskip
\noindent
{\bf \boldmath$k$-in-a-Subgraph.}
The problem of detecting an induced subgraph of a certain type containing a set of $k$ specified vertices (called terminals as well) has also been 
extensively studied, in particular for the case where the induced subgraph is required to be a tree, cycle, or path. Then this problem is called  {\sc $k$-in-a-Tree}, {\sc $k$-in-a-Cycle}, and {\sc $k$-in-a-Path}, respectively.
Derhy and Picouleau~\cite{derhy2009}  proved that {\sc $k$-in-a-Tree} is \classNP-complete when $k$ is part of the input even for planar bipartite cubic graphs. Chudnovsky and Seymour~\cite{CS10}  proved that {\sc $3$-in-a-Tree} is polynomial-time solvable.
Very recently, Lai, Lu and Thorup~\cite{LLT19} gave a faster, near-linear algorithm for {\sc $3$-in-a-Tree} and
Gomes et al.~\cite{GSSS}  proved various parameterized complexity results for {\sc $k$-in-a-Tree}.
However, the (classical) complexity of {\sc $k$-in-a-Tree} is still open for every fixed $k\geq 4$.
In contrast, the problems {\sc $2$-in-a-Cycle} and {\sc $3$-in-a-Path} are \classNP-complete; this follows from the aforementioned results of Fellows~\cite{Fe89} and  Bienstock~\cite{Bi91}.

We note that for every fixed $k$, the problems {\sc $k$-in-a-Cycle}, {\sc $k$-in-a-Path} and {\sc $k$-in-a-Tree} can be reduced to a polynomial number of instances of {\sc Induced Disjoint Paths} due to the relaxed conditions~2 and~3 in the problem definition (see, for example,~\cite{RTV19} for details). Hence, all the aforementioned polynomial-time results for
the latter problem restricted to special graph classes carry over to these three $k$-in-a-subgraph problems for every fixed integer $k\geq 1$. 
It is also known that {\sc $k$-in-a-Tree} can be solved in polynomial time for graphs of girth at least $k$~\cite{LT10} even if $k$ is part of the input (see~\cite{DPT09} for the case where $k=3$).
L\'{e}v\^{e}que et al.~\cite{LLMT09} proved that {\sc $2$-in-a-Cycle} is polynomial-time solvable for graphs not containing an induced path or induced subdivided claw on some fixed number of vertices. 
Radovanovi\'c, Trotignon and Vu\v{s}kovi\'c~\cite{RTV19} proved that {\sc $k$-in-a-Cycle} is fixed-parameter tractable for 
(theta,wheel)-free graphs when parameterized by $k$. 

In our paper we add to these positive results by showing 
that the three problems {\sc $k$-in-a-Tree}, {\sc $k$-in-a-Cycle}, and {\sc $k$-in-a-Path}
are polynomial-time solvable on AT-free graphs even when $k$ is part of the input. Note that solving {\sc $k$-in-a-Cycle} in polynomial time is straightforward, because AT-free graphs do not contain induced cycles on six or more vertices. For the remaining two problems, we again apply our main result and make explicit use of conditions 2 and~3 for the terminals in the definition of {\sc Induced Disjoint Paths}. 
The standard brute-force approach that reduces the instance to a polynomial number of instances of {\sc Induced Disjoint Paths} 
will only yield this result for any fixed integer~$k$, because in the worst case it must process all $k!$ orderings of the terminals. However, we are able to give a more direct approach for AT-free graphs, 
based on dynamic programming, that works even when $k$ is part of the input.

\section{Preliminaries}\label{sec:prelim}

Throughout the paper, we consider only finite, undirected graphs without multiple edges and self-loops. We denote the vertex set and edge set of a graph $G$ by $V_G$ and $E_G$, respectively; we may omit subscripts if no confusion is possible.

Let $G=(V,E)$ be a graph.
For $U\subseteq V$, the graph~$G[U]$ denotes the subgraph of $G$ induced by the vertices in $U$.
We denote the (open) neighbourhood of a vertex~$u$ by $N(u)=\{v\; |\; uv\in E\}$ and its closed neighbourhood by $N[u]=N(u)\cup \{u\}$.
We denote the (open) neighbourhood of a set $U\subseteq V$ by $N(U)=\{v\in V\setminus U\; |\;  uv\in E\; \mbox{for some}\; u\in U\}$ and its 
closed neighbourhood by $N[U]=N(U)\cup U$.
We let $d(u)=|N(u)|$ denote the {\it degree} of a vertex~$u$. Whenever it is not clear from the context, we may add an extra subscript $_G$ to these notations.

The {\it length} of a path is its number of edges.
Paths $P_1,\ldots,P_k$ in a graph~$G=(V,E)$ are {\it mutually induced} if any two distinct $P_i$ and $P_j$  have neither common vertices nor adjacent vertices (except perhaps their end-vertices). 
Let $uw\in E$. 
The \emph{edge subdivision} of $uw$ removes $uw$ and adds a new vertex~$v$ with edges $uv$ and $vw$.
A graph~$G'$ is a \emph{subdivision} of $G$ if $G'$ can be obtained from $G$ by a sequence of edge subdivisions.
A graph~$H$ is an \emph{induced topological minor} of $G$, if  $G$ has an induced subgraph that is isomorphic to a subdivision of $H$.

Let $G$ be a graph in which we specify $r$ distinct vertices $p_1,\ldots,p_r$. Let $H$ be a graph in which we specify $r$ distinct vertices $q_1,\ldots,q_r$.
Then $G$ contains $H$ as an induced topological minor {\it anchored} in $(p_1,q_1),\ldots,(p_k,q_k)$ if
$G$ contains an induced subgraph isomorphic to a subdivision of $H$ such that the isomorphism maps $p_i$ to $q_i$ for $i=1,\ldots, r$. 
The graphs $G$ and $H$ may have common vertices. If $p_i=q_i$ for $i=1,\ldots,r$, we speak of ``being anchored in $p_1,\ldots,p_k$'' instead.

A subset $U\subseteq V$ {\it dominates} $G$ if $u\in U$ or $u\in N(U)$ for each $u\in V$.
A pair of vertices $\{x,y\}$  is a \emph{dominating pair} of a graph~$G$ if the vertex  set of every $(x,y)$-path  dominates $G$. A path $P$ dominates a vertex~$u$ if $u\in N[V_P]$. A path $P$ dominates a vertex set $U$ if it dominates each $u\in U$. 
Corneil, Olariu and Stewart proved the following result, which we will use as a lemma.

\begin{lemma}[\cite{CorneilOS97,CorneilOS99}]\label{l-cc}
Every connected AT-free graph has a dominating pair and such a pair can be found in linear time.
\end{lemma}

Throughout the paper,
 we let $T=\bigcup_{i=1}^k \{s_{i},t_{i}\}$ denote the set of terminals, whereas $S=\{(s_1,t_1),\ldots,(s_k,t_k)\}$ denotes the set of terminal pairs.  
Note that a vertex~$u$ may be in more than one terminal pair. 
We say that a terminal~$s_i$ or $t_i$ is {\it represented} by a vertex~$u$ if $u=s_i$ or $u=t_i$, respectively. We also say that $u$ is a {\it terminal vertex}.
If a vertex~$u$ does not represent any terminal, then we call $u$ a {\it non-terminal vertex}.

\section{Induced Disjoint Paths}\label{s-dip}

In this section we present our polynomial-time algorithm for the {\sc Induced Disjoint Paths} problem. Our algorithm has as input a pair $(G,S)$, where $G$ is an AT-free graph and $S$ is a set of terminal pairs.
It consists of the following three phases.

\medskip
\noindent
{\bf Phase 1.} Preprocess $(G,S)$ to derive a number of convenient properties. For instance, afterwards two terminals of the same pair are non-adjacent. 

\medskip
\noindent
{\bf Phase 2.}  Derive a number of  structural properties of $(G,S)$. The algorithm constructs an auxiliary graph~$H$, which is obtained from the subgraph of $G$ induced by the terminal vertices by adding a path of length~2 between each pair of terminals. It then checks whether $H$
satisfies some necessary conditions, such as being AT-free and being an anchored topological minor of $G$.  The latter condition is also shown to be sufficient and demands the construction of another auxiliary graph in order to describe how induced paths connecting terminal pairs 
may interfere with each other.

\medskip
\noindent
{\bf Phase 3.} Perform dynamic programming using the information from 
Phases~1--2.

\subsection{Phase 1: Preprocessing}\label{sec:preproc}

Let $(G,S)$ be an instance of {\sc Induced Disjoint Paths}, where $G=(V,E)$ is an AT-free graph on $n$ vertices with a set $S$ of $k$ terminal pairs $(s_1,t_1),\ldots, (s_k,t_k)$. 
Recall that we assume that the terminal pairs in $S$ are pairwise distinct if we consider them to be unordered.
In other words, although a vertex may represent more than one terminal, we do not allow two terminal pairs to coincide, that is,  $\{s_i,t_i\}\neq\{s_j,t_j\}$ for all $i\neq j$.
We can safely make this assumption, as the algorithm for detecting induced topological minors that uses our algorithm for {\sc Induced Disjoint Paths} as a subroutine does not require two terminal pairs to coincide, as we shall see.  It may happen though  that in this application two terminals of the same pair are the same. However,
we will show that we can easily work around this, and to avoid any further technicalities in our algorithm for {\sc Induced Disjoint Paths}
we assume from now on that this is not the case either. Both assumptions are summarized in the following condition.

\medskip
\noindent
{\bf Condition i)}
For all $i\neq j$, $\{s_i,t_i\}\neq\{s_j,t_j\}$, and for all $i$, $s_i\neq t_i$.

\medskip
\noindent
After imposing Condition~i) on $S$, we continue with 
the following preprocessing steps. We first apply Step~1, then Step~2, and then Step~3, where we perform each step as long as possible.
In the below, $M(u)$ denotes the set of all terminals that form a terminal pair with vertex~$u$.
Recall also that $T=\bigcup_{i=1}^k\{s_i,t_i\}$ denotes the set of terminals. 

\medskip
\noindent{\bf Step~1.} For all $u,v\in T$ with $uv\in E$, remove all common neighbours of $u$ and $v$ that are not terminals from $G$, that is,
remove all vertices of $N(u)\cap N(v)\cap (V\setminus T)$.

\medskip
\noindent{\bf Step~2.} 
Let $U$ denote the set of all terminal vertices $u$ such that $u$ only represents terminals whose partners are in $N(u)$; that is, $U = \{ u \in T \mid M(u)\subseteq N(u) \}$.
Remove $U$ and all non-terminal vertices of $N(U)$ from $G$. Remove also all terminal pairs $(s_j,t_j)$ from $S$ for which $s_j \in U$ or $t_j \in U$.

\medskip
\noindent{\bf Step~3.} For all $(s_i,t_i)$ with $s_it_i\in E$, remove $(s_i,t_i)$ from $S$.

\medskip
\noindent
Observe that Steps 1--3 preserve AT-freeness.
Steps 1--3 are identical to 
Rules~1--3 
in~\cite{GPV12-esa} for claw-free graphs. As such, the correctness of these three steps, as stated in our next lemma, can be proved by using exactly the same arguments as in~\cite{GPV12-esa}. 

\begin{lemma}\label{lem:p1-3}
Applying Steps~1--3 takes polynomial time and results in a new instance $(G',S')$, where $G'$ is an induced (and hence AT-free) 
subgraph of $G$ and $S'\subseteq S$, such that $(G',S')$  
is a yes-instance  of {\sc Induced Disjoint Paths} if and only if $(G,S)$ is a yes-instance of {\sc Induced Disjoint Paths}.
\end{lemma}

For convenience we denote the obtained instance by $(G,S)$ as well, and we also assume that $|S|=k$.  

\begin{figure}[ht]
\centering\scalebox{0.75}{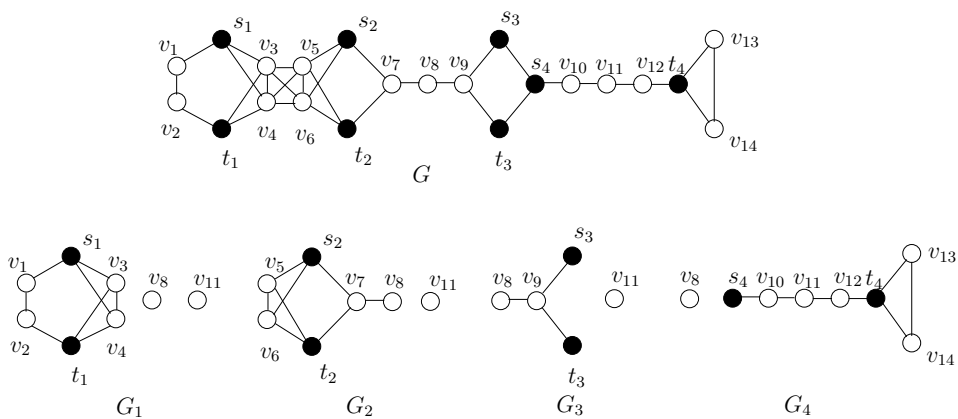}
\caption{A graph $G$ with four terminal pairs and the corresponding graphs $G_1,\ldots,G_4$. 
\label{fig:Gi}}
\end{figure}

For $i=1,\ldots,k$, let $G_i$ denote the subgraph obtained from $G$ after removing all terminal vertices not equal to $s_i$ or $t_i$, together with all of their 
neighbours  not equal to $s_i$ or $t_i$ (should $s_i$ or $t_i$ be adjacent to a terminal from some other pair), that is, $G_i$ is the subgraph of $G$
induced by $(V\setminus(\bigcup_{v\in T\setminus \{s_i,t_i\}}N[v]))\cup\{s_i,t_i\}$ (see Fig.~\ref{fig:Gi}).
The following lemma is straightforward to see.

\begin{lemma}\label{lem:cut}
If $s_i$ and $t_i$ are in different connected components of $G_i$ for some $1\leq i\leq k$, then $(G,S)$ is a no-instance of {\sc Induced Disjoint Paths}.
\end{lemma}
Hence, we can add the following preprocessing step, which we can perform in polynomial time.

\medskip
\noindent{\bf Step~4.} If some $s_i$ and $t_i$ are in two different connected components of  
$G_i$, then return {\tt no}.

\medskip
\noindent
Summarizing, applying Steps~1--4 takes polynomial time and results in the 
following additional conditions for our instance:

\medskip
\noindent
{\bf Condition ii)}  The terminals of each pair $(s_i,t_i)$ are not adjacent.

\medskip
\noindent
{\bf Condition iii)}
 The terminals of each pair $(s_i, t_i)$ are in the same connected component of $G_i$.

\subsection{Phase 2: Obtaining Structural Results}\label{sec:struct}

Let $G$ be an AT-free graph with a set $S=\{(s_1,t_1),\ldots,(s_k,t_k)\}$ of terminal pairs  satisfying Conditions~i)--iii). 
In $G[T]$,
there is no edge between any two terminals of the same pair due to Condition~ii).
We extend $G[T]$
by joining the terminals of each pair $(s_i,t_i)$ by a path $P_i$ of length~2, that is, for each pair $(s_i,t_i)$ we introduce a new vertex that we make adjacent only to 
$s_i$ and $t_i$. We denote the resulting graph by $H$; see Fig.~\ref{fig:GH} for an example.
Note that $G[T]$ is an induced subgraph of $H$.
The inner vertex of each $P_i$ is called a \emph{path-vertex}, and the two edges of each $P_i$ are called \emph{path-edges}.

\begin{figure}[ht]
\centering\scalebox{0.75}{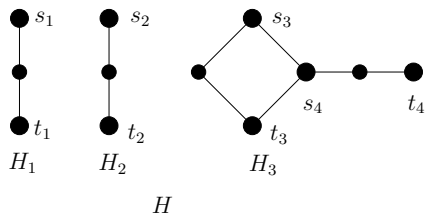}
\caption{The graph $H$ that is constructed from the graph~$G$ shown in Fig.~\ref{fig:Gi}. 
\label{fig:GH}}
\end{figure}

We now prove the following lemma; note that in this lemma $H'=H$ is possible.

\begin{lemma}\label{lem:H}
The pair $(G,S)$ is a yes-instance of {\sc Induced Disjoint Paths} if and only if 
$G$ contains an induced subgraph that is isomorphic to $H'$, where $H'$ is obtained from $H$ after subdividing a number of path-edges one or more times, such that the isomorphism maps every terminal of $G$ to the same terminal of $H$.
\end{lemma}

\begin{proof}
Suppose that $(G,S)$ is a yes-instance of {\sc Induced Disjoint Paths} and let $P_1,\ldots,P_k$ be a solution to $(G,S)$. Observe that the only terminal vertices to which any internal vertices of a path $P_i$ are adjacent are the terminal vertices that represent $s_i$ and $t_i$. Hence, the union of $G[T]$ and the paths $P_1,\ldots,P_k$ forms a graph $H'$ that satisfies the conditions of the lemma. The converse is readily seen.
\qed\end{proof}

Lemma~\ref{lem:H} implies that the pair $(G,S)$ is a yes-instance of {\sc Induced Disjoint Paths} if and only if $G$ contains $H$ as an induced topological minor anchored in the terminals of~$T$.
Moreover, Lemma~\ref{lem:H} immediately implies the next lemma.

\begin{lemma}\label{lem:H-AT}
If  $(G,S)$ is a yes-instance of {\sc Induced Disjoint Paths} then $H$ is AT-free.
\end{lemma}

Lemma~\ref{lem:H-AT} yields the following step,
which can be straightforwardly implemented in polynomial time by checking for the existence of an asteroidal triple.

\medskip
\noindent{\bf Step~5.} If $H$ is not an AT-free graph, then return {\tt no}.

\medskip
\noindent
From now on, we assume that $H$ is AT-free. We deduce the following two lemmas.

\begin{lemma}\label{lem:triangle}
The graph~$G[T]$ is triangle-free.
\end{lemma}

\begin{proof}
In order to obtain a contradiction, assume that $G[T]$ contains a triangle. Then $H[T]$ also contains this triangle.
By Condition~ii), each vertex of the triangle is contained in at least one terminal pair that contains a vertex outside the triangle. Then let $(s_i,t_i), (s_{j},t_{j}), (s_{r},t_{r})$ be three terminal pairs for which, without loss of generality, $s_{i},s_{j},s_{r}$ induce the triangle. By Condition~ii), $s_i,s_j,s_r$ are not adjacent to $t_i,t_j,t_r$, respectively.
In $H$, let  $x_i,x_j,x_r$ be the path-vertices of the paths $P_i,P_j,P_r$, respectively. Recall that path-vertices
have degree~2.
Then $x_i,x_j,x_r$ form an asteroidal triple in $H$, a contradiction to the AT-freeness of $H$.
\qed 
\end{proof}

\begin{lemma}\label{lem:degree}
Every vertex~$u\in T$ is included in at most five terminal pairs.
\end{lemma}

\begin{proof}
In order to obtain a contradiction, assume that $T$ contains a vertex~$u$ that is included in six terminal pairs $(s_{i_1},t_{i_1}),\ldots,(s_{i_6},t_{i_6})$. 
We assume without loss of generality that
$u=s_{i_1}=\ldots=s_{i_6}$. 
Among the vertices $t_{i_1},\ldots,t_{i_6}$, it follows from the Ramsey number $R(3,3)=6$ that we have either three pairwise adjacent or three pairwise non-adjacent vertices. In the first case, $G[T]$ contains a triangle. However, this is not possible  due to Lemma~\ref{lem:triangle}. In the second case, 
let  $t_{i_1},t_{i_2},t_{i_3}$ denote the three pairwise non-adjacent vertices. 
Then we find that $t_{i_1},t_{i_2},t_{i_3}$ form an asteroidal triple in $H$, because  $t_{i_1},t_{i_2},t_{i_3}$ are joined to $u$ by three mutually induced paths $P_{i_1},P_{i_2},P_{i_3}$, a contradiction to the AT-freeness of $H$. 
\qed 
\end{proof}

Let $H_1,\ldots,H_r$ be the connected components of $H$. 
We observe that the terminals of each pair are in the same connected component of $H$ due to the paths $P_i$.
Hence, we can define the set $S_i\subseteq S$ of terminal pairs in a connected component $H_i$. We write $T_i=V_{H_i}\cap T$ to denote the set of terminals in $H_i$.

\begin{lemma}\label{lem:DP}
Each $H_i$ has a dominating pair $\{x_i,y_i\}$ with $x_i,y_i\in T_i$. Moreover, such a dominating pair can be found in linear time.
\end{lemma}

\begin{proof}
Each connected
component $H_i$ is a connected AT-free graph. Hence, $H_i$ has a dominating pair of vertices due to Lemma~\ref{l-cc}. Moreover, by the same lemma, there exists a linear-time algorithm that finds such a dominating pair $\{u,v\}$. 
If both $u$ and $v$ are in $T_i$, then we are done.
Suppose that this is not the case, say $u\notin T_i$. Then
$u$ must be the path-vertex of some $(s_j,t_j)$-path $P_j$ in $H_i$.

We claim that $\{s_j,v\}$ or else $\{t_j,v\}$ is a dominating pair. This can be seen as follows.
If $\{s_j,v\}$ is not a dominating pair, then there is an $(s_j,v)$-path $Q_1$ that does not dominate $H_i$. 
This path $Q_1$ does not use vertex~$u$; otherwise it would contain a subpath from $u$ to $v$, which dominates
$H_i$ because $\{u,v\}$ is a dominating pair. Hence, we can extend $Q_1$ by $u$ to obtain the path $uQ_1$.
This path is an $(u,v)$-path. Consequently, $uQ_1$ dominates $H_i$. As $u$ is only adjacent to $s_j$ and $t_j$, and $Q_1$ does not dominate $H_i$, we find that $Q_1$ does not dominate $t_j$. In other words, 
$Q_1$ avoids $t_j$ and the neighbourhood of $t_j$. By the same arguments, there exists a $(t_j,v)$-path $Q_2$ that avoids $s_j$ and the neighbourhood of $s_j$. Condition~ii) tells us that  $s_j$ and $t_j$ are not adjacent. As $v\in V_{Q_1}\cup V_{Q_2}$, we find that $v$ is not adjacent to $s_j,t_j$. Moreover, $u$ is only adjacent to $s_j$ and $t_j$.  Hence, $s_j,t_j,v$ form an asteroidal triple, a contradiction to the AT-freeness of $H_i$.  

From the above, it is clear that we can decide in linear time if $\{s_j,v\}$ or $\{t_j,v\}$ is a dominating pair.
We only have to check whether there exists an $(s_j,v)$-path that avoids $t_j$ and the neighbourhood of $t_j$. If not, then $\{s_j,v\}$ is a dominating pair, and otherwise $\{t_j,v\}$ is a dominating pair. If $v\notin T_i$, then $v$ is a path-vertex, and we can repeat the same arguments as we used for $u$ to replace $v$ by one of its neighbours in $T_i$. This competes the proof of Lemma~\ref{lem:DP}.
\qed 
\end{proof}

For the dominating pairs $\{x_i,y_i\}$ found by Lemma~\ref{lem:DP},
we compute in linear time a shortest $(x_i,y_i)$-path $D_i$ in $H_i$ for $i=1,\ldots,r$. 
The next two lemmas 
show a number of properties of these paths $D_i$.

\begin{lemma}\label{lem:Pi-term}
Each $D_i$ contains at least one of the terminals of every pair in $S_i$.
\end{lemma}

\begin{proof}
The path $D_i$ dominates all vertices of $H_i$. In particular, it dominates the path-vertex of the $(s_j,t_j)$-path for each $(s_j,t_j)\in S_i$. 
This vertex is only adjacent to $s_j$ and $t_j$.
Hence, $s_j\in V_{D_i}$ or $t_j\in V_{D_i}$. \qed
\end{proof}

\begin{lemma}\label{lem:Pi-deg}
Each vertex of $D_i$ is adjacent to at most five path-vertices of $H_i$ that are not on $D_i$ and to at most two terminals that are not on $D_i$.
\end{lemma}

\begin{proof}
Let $u$ be a vertex of $D_i$. If $u$ is not a terminal, then $u$ is a path-vertex, and consequently, has degree~2. Hence, we may assume that $u$ is a terminal.
Lemma~\ref{lem:degree} tells us that $u$ can represent at most five terminals. This implies that $u$ is adjacent to at most five path-vertices not on $D_i$.  

We now show that $u$ is adjacent to at most two terminals not on $D_i$. To obtain a contradiction, suppose that $u$ is adjacent to three terminals. 
Since at least one terminal of each terminal pair belongs to $D_i$ due to Lemma~\ref{lem:Pi-term}, these three terminals are from three different pairs. Hence we may without loss of generality denote them by
$s_{i_1},s_{i_2},s_{i_3}$.
Let $v,w,z$ be the path-vertices of the $(s_{i_1},t_{i_1})$-, $(s_{i_2},t_{i_2})$- and
$(s_{i_3},t_{i_3})$-paths $P_{i_1},P_{i_2},P_{i_3}$,  respectively. 
Note that $u\notin \{t_{i_1},t_{i_2},t_{i_3}\}$, because $H_i$ contains no edge between two terminals of the same pair
by Condition~(ii).
This means that no vertex of $\{v,w,z\}$ is adjacent to $u$.
Hence $v,w,z$ form an asteroidal triple, a contradiction to the AT-freeness of $H_i$. \qed
\end{proof}

Recalling Lemma~\ref{lem:H} and using Lemma~\ref{lem:Pi-term}, we obtain the following.

\begin{lemma}\label{lem:subdiv}
Let $H_i'$ be an AT-free graph obtained from $H_i$ by 
subdividing a number of path-edges one or more times.
Let $P_j'$ and $D_i'$ be the resulting paths obtained from the paths $P_j$ and $D_i$, respectively.
Then,
\begin{itemize} 
\item[a)] the length of each $P_j'$ that is not a subpath of $D_i'$ is at most~$3$ (implying that every internal vertex of $P_j'$ is  
adjacent to at least one of $s_j,t_j$); 
\item[b)] $D_i'$ dominates all but at most two vertices of $H_i'$.
\end{itemize}  
\end{lemma}

\begin{proof} First, we prove \emph{a)}. 
In order to obtain a contradiction, assume that $P_j'$ has length at least~4. Let $u$ be an internal vertex of $P_j'$ that is at distance at least~2 from $s_j$ and $t_j$. 
By Lemma~\ref{lem:Pi-term}, at least one of the vertices $s_j,t_j$ is a vertex of $D_i'$. Assume that $s_j\in V_{D_i'}$. Then either $t_j\in V_{D_i'}$ or $t_j$ is adjacent to a vertex~$z\in V_{D_i'}$.

First suppose that $t_j\in V_{D_i'}$. Let $Q$ be an $(s_j,t_j)$-subpath of $D_i'$.
Then $Q$ has length at least~2, because Condition~ii) tells us that $s_j$ and $t_j$ are not adjacent in $G$, and consequently, not in $H'$. Since $u$ has no neighbours outside $P_j'$, we find that
$Q$ avoids the neighbourhood of $u$. 
However, then $u,s_j,t_j$ form an asteroidal triple, a contradiction to the AT-freeness of $H_i$.   

Now suppose that $t_j$ is adjacent to a vertex~$z\in V_{D_i'}$. We choose $z$ in such a way that
$t_j$ is not adjacent to any vertex on the subpath $Q$ of $D_i'$ from $s_j$ to $z$.
The path $Qt_j$ has length at least~2, because $s_j$ and $t_j$ are not adjacent in $H$.
Moreover, $Qt_j$ avoids the neighbourhood of $u$.
Hence, $u,s_j,t_j$ form an asteroidal triple, a contradiction to the AT-freeness of $H$.

To prove \emph{b)}, assume that $D_i'$ does not dominate three vertices $u_1,u_2,u_3\in V_{H_i}$. 
Since we only subdivide path-edges and $D_i$ dominates each terminal, 
we find that $u_1$, $u_2$, $u_3$ are internal vertices of some $(s_j,t_j)$-paths $P_j'$. 
By Lemma~\ref{lem:Pi-term}, at least one of the vertices $s_j,t_j$ is a vertex of $D_i'$ for each such $P_j'$, that is, $D_i'$ dominates at least one internal vertex of $P_j'$. Then by \emph{a)}, each $P_j'$ can have at most one internal vertex not dominated by $D_i'$.  
Hence,  $u_1,u_2,u_3$ are internal vertices of distinct paths $P_{j_1}',P_{j_2}',P_{j_3}'$, respectively. We assume without loss of generality that $s_{j_1},s_{j_2},s_{j_3}\in V_{D_i'}$. Then the path composed of the $(u_1,s_{j_1})$-subpath of $P_{j_1}'$, the $(s_{j_1},s_{j_2})$-subpath of $D_{i}'$, and the $(s_{j_2},u_2)$-subpath of $P_{j_2}'$ is a $(u_1,u_2)$-path 
in $H_i'$
that avoids the neighbourhood of $u_3$. By the same arguments, $H_i'$ has a $(u_1,u_3)$-path  that avoids the neighbourhood of $u_2$ and a $(u_2,u_3)$-path 
that avoids the neighbourhood of $u_1$. We conclude that $u_1,u_2,u_3$ form an asteroidal triple 
in $H_i'$, a contradiction to the AT-freeness of 
$H_i'$.
This completes the proof of Lemma~\ref{lem:subdiv}.\qed
\end{proof}

\begin{figure}[ht]
\centering\scalebox{0.75}{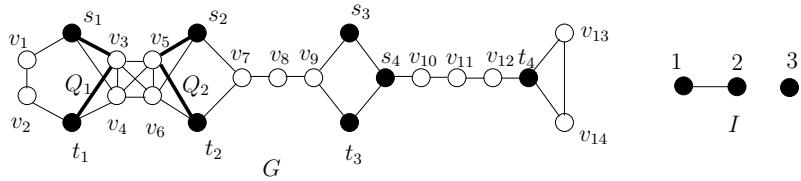}
\caption{The graph $I$ that is constructed from the graph  $G$ shown in Fig.~\ref{fig:Gi}. The terminal pairs $(s_1,t_1)$ and $(s_2,t_2)$ of $G$ are interfering. Note  that there is no interference between $(s_2,t_2)$ and $(s_3,t_3)$.  
\label{fig:interf}}
\end{figure}

Recall that Condition~iii) tells us that the terminals  $s_i$ and $t_i$ are in the same connected component of the graph 
$G_i$ for $i=1,\ldots, k$. 
Two terminal pairs $(s_i,t_i)$ and $(s_j,t_j)$ are {\it interfering}  if there is an induced $(s_i,t_i)$-path $Q_i$ in $G_i$ and an induced $(s_j,t_j)$-path $Q_j$ in $G_{j}$, respectively, such that $Q_i$ and $Q_j$ are not mutually induced; see Fig.~\ref{fig:interf} for an example.  
We say that there is \emph{interference} between two 
sets of terminal pairs $S_i$ and $S_j$ if  a terminal pair from $S_i$ and a terminal pair from $S_j$ are interfering.

\begin{lemma}\label{lem:interf}
Let $(s_i,t_i)$ and $(s_j,t_j)$ be interfering terminal pairs  from two different connected components of $H$. Let $Q_i$ and
$Q_j$ be induced $(s_i,t_i)$- and $(s_j,t_j)$-paths 
in $G_{i}$ and $G_j$, respectively, such that $Q_i$ and $Q_j$ are not mutually induced.
Then $Q_i$ and $Q_j$ are vertex-disjoint.
Moreover, for any edge $uv\in E_G$ with $u\in V_{Q_i}$ and $v\in V_{Q_j}$, 
both $u\in N(s_i)\cup N(t_i)$ and $v\in N(s_j)\cup N(t_j)$ hold. 
\end{lemma}

\begin{proof}
To obtain a contradiction, we first assume that $Q_i$ and $Q_j$ have a common vertex~$w$. As 
$(s_i,t_i)$ and $(s_j,t_j)$ are terminal pairs from two different connected components of $H$, 
the four vertices $s_i,t_i,s_j,t_j$ are pairwise distinct.
As $Q_i$ is a path in $G_i$, we find that $Q_i$ avoids $N[s_j]$ and $N[t_j]$. Symmetrically, 
$Q_j$ avoids $N[s_i]$ and $N[t_i]$.
This means that $w$ is not adjacent to any of the vertices $s_i,t_i,s_j,t_j$. 
However, then $s_i,t_i,s_j$ form an asteroidal triple, a contradiction to the AT-freeness of $H$.

Now assume that there is an edge $uv\in E_G$ such that $u\in V_{Q_i},v\in V_{Q_j}$ and
$u\notin N(s_i)\cup N(t_i)$. 
Recall that $Q_i$ avoids $N[s_j]$ and $N[t_j]$, and that $Q_j$ avoids $N[s_i]$ and $N[t_i]$.
Then we obtain a contradiction, because $s_i,t_i,s_j$ again form an asteroidal triple.\qed
\end{proof}

Lemma~\ref{lem:interf} implies that any interference between two terminal pairs $(s_i,t_i)$ and $(s_j,t_j)$ from different connected components of $H$ stems from an edge whose endpoints are in 
$N(s_i)\cup N(t_i)$ and $N(s_j)\cup N(t_j)$, respectively. 
We now obtain the following lemma.

\begin{lemma}\label{lem:interf-poly}
It is possible to check in polynomial time whether two terminal pairs  $(s_i,t_i)$ and $(s_j,t_j)$ from different connected components of $H$ are interfering.
\end{lemma}

\begin{proof}
By Lemma~\ref{lem:interf}, we only need to check if there exists four vertices $u_1\in N(s_i)$, $u_2\in N(t_i)$, $v_1\in N(s_j)$, 
$v_2\in N(t_j)$ such that 
\begin{itemize}
\item there exists a path $Q_i$ in $G_i$ from $u_1$ to $u_2$ that does not contain $s_i$ or $t_i$;
\item there exists a path $Q_j$ in $G_j$ from $v_1$ to $v_2$ that does not contain $s_j$ or $t_j$; and
\item there exists an edge $uv$ in $G$ with $u\in \{u_1,u_2\}$ and $v\in \{v_1,v_2\}$.
\end{itemize}
For each 4-tuple $(u_1,u_2,v_1,v_2)$, checking the above three conditions can be done in polynomial time. \qed
\end{proof}

Lemma~\ref{lem:interf-poly} enables us to construct in polynomial time an auxiliary graph~$I$ with vertices $1,\ldots,r$ and edges $ij$ if and only if there is interference between $S_i$ and $S_j$; see Fig.~\ref{fig:interf} for an example. This leads to the following lemma.

\begin{lemma}\label{lem:F}
The graph~$I$ is a disjoint union of paths.
\end{lemma}

\begin{proof}
First, we prove that $d_I(i)\leq 2$ for all $1\leq i\leq r$. To obtain a contradiction, assume that $d_I(i)\geq 3$. Let $j_1,j_2,j_3\in\{1,\ldots,r\}$ be three vertices of $I$ adjacent to $i$.
For any pair of terminals $(s_p,t_p)$ from $S_i\cup S_{j_1}\cup S_{j_2}\cup S_{j_3}$, there is an $(s_p,t_p)$-path $R_p$ in $G_p$ by Condition~iii).
Moreover, for each $l\in\{1,2,3\}$, there are interfering terminal pairs $(s_{i_l},t_{i_l})$ and $(s_{j_l},t_{j_l})$ in 
$S_i$ and $S_{j_l}$ respectively.
Hence, by definition, there 
exist an $(s_{i_l},t_{i_l})$-path $Q_{i_l}$ in 
$G_i$
and an $(s_{j_l},t_{j_l})$-path  $Q_{j_l}$ in $G_{j_l}$, respectively,
 such that $Q_{i_l}$ and $Q_{j_l}$ are not mutually induced. 

By definition of the graphs $G_i$,
the paths $R_p$, $Q_{i_l}$, and $Q_{j_l}$ each avoid the closed neighbourhood of terminal vertices that are not their endpoints. In particular, for any pair $S_a,S_b$ among $S_i$, $S_{j_1}$, $S_{j_2}$, and $S_{j_3}$, the paths among $R_p$, $Q_{i_l}$, and $Q_{j_l}$ whose endpoints are in $T_a$ avoid the closed neighbourhood of all terminal vertices in $T_b$, because $H_i$, $H_{j_1}$, $H_{j_2}$, and $H_{j_3}$ are connected components of $H$.
 
Let $R$ be the subgraph of $G$ induced by the vertices of all the paths $R_p$, $Q_{i_l}$ and $Q_{j_l}$. 
Let $u_1,u_2,u_3$ be arbitrary terminal vertices from $T_{j_1},T_{j_2},T_{j_3}$ respectively, and let $v$ be a terminal from $T_i$. 
For each $l\in\{1,2,3\}$, there is a $(v,u_l)$-path in $R$ that avoids the neighbourhoods of terminals $u_p$ for $p\in\{1,2,3\}$, $p\neq l$. 
In particular, recalling the observations of the previous paragraph: 
\begin{itemize}
\item the connectedness of $H_i$ implies that using the edges of the paths $R_p$ there is a path in $R$ from $v$ to $s_{i_l}$ that avoids the closed neighbourhood of any terminal vertex in $T_{j_1},T_{j_2},T_{j_3}$;
\item the interference of $Q_{i_l}$ and $Q_{j_l}$ implies there is a path in $R$ from $s_{i_l}$ to $s_{j_l}$ that avoids the closed neighbourhood of any terminal vertex in $T_{j_{l'}}$ for $l' \not= l$;
\item the connectedness of $H_{j_l}$ implies that using the edges of the paths $R_p$ there is a path in $R$ from $s_{j_l}$ to $u_l$ that avoids the closed neighbourhood of any terminal vertex in $T_{j_{l'}}$ for $l' \not= l$ and in $T_i$.
\end{itemize}
Hence, $u_1,u_2,u_3$ form an asteroidal triple, a contradiction.

It remains to prove that $I$ has no cycles. 
In order to obtain a contradiction, assume that $i_0i_1\ldots i_q$ with $i_0=i_q$ is a cycle in $I$. 
For any pair of terminals $(s_p,t_p)$ from $S_{i_1}\cup\ldots\cup S_{i_q}$,  there is an $(s_p,t_p)$-path $R_p$ in 
$G_p$ by Condition~iii).
Moreover, for each $l\in\{1,\ldots,q\}$, there are interfering terminal pairs $(s_{l_1},t_{l_1})$ and $(s_{l_2},t_{l_2})$ in 
$S_{i_{l-1}}$ and $S_{i_l}$ respectively.
Hence, by definition, there
exists an $(s_{l_1},t_{l_1})$-path~$Q_{l_1}$ in $G_{l_1}$ and an $(s_{l_2},t_{l_2})$-path~$Q_{l_2}$ in $G_{l_2}$, respectively,
such that $Q_{l_1}$ and $Q_{l_2}$ are not mutually induced.

Let $R$ be the subgraph of $G$ induced by the vertices of all the paths $R_p$, $Q_{l_1}$ and $Q_{l_2}$. 
Let $u_1,\ldots,u_q$ be arbitrary terminal vertices from $H_{i_1},\ldots,H_{i_q}$, respectively. Observe that for each $j\in\{1,\ldots,q\}$, there is a $(u_{j-1},u_j)$-path in $R$ that avoids the neighbourhood of $u_l$ for each $l\in\{1,\ldots,q\}$, $l\neq j-1,j$. 
The existence of such a path follows along the same lines as before.
Hence, $u_1,u_2,u_3$ form an asteroidal triple, a contradiction.
\qed 
\end{proof}

\begin{figure}[ht]
\centering\scalebox{0.75}{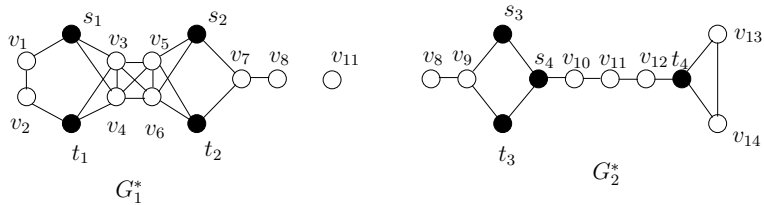}
\caption{The graphs $G_1^*$ and $G_2^*$ that are constructed from the graph~$G$ shown in Fig.~\ref{fig:Gi}.   
\label{fig:star}}
\end{figure}

Let $I_1,\ldots,I_l$ be the connected components of $I$, and let  $J_1,\ldots,J_l$ be their vertex sets, respectively. 
For $h=1,\ldots,l$, we define $X_h=\{(s_j,t_j) \;|\; (s_j,t_j)\in \bigcup_{p\in J_h}S_p\}$, and we 
let $G_h^*$ be the graph obtained from $G$ by 
removing all vertices of the 
closed neighbourhoods of all the terminals that are not included in $X_h$; see Fig~\ref{fig:star} for an example.
Lemma~\ref{lem:F-comp} shows how $G$ is related to the graphs $G_h^*$. 

\begin{lemma}\label{lem:F-comp}
The instance $(G,S)$ is a yes-instance of {\sc Induced Disjoint Paths} if and only if  $(G_h^*,X_h)$ is a yes-instance of {\sc Induced Disjoint Paths} for all $1\leq h\leq l$.
\end{lemma}

\begin{proof}
Clearly, a solution for $(G,S)$ induces a solution for the instances $(G_h^*,X_h)$.
Suppose that we have a solution for each $(G_h^*,X_h)$.  Then two paths from two different solutions are mutually induced, because no two
terminal pairs from two different sets $X_h$ and $X_{h'}$ are interfering.
Hence, the union of these solutions for the instances $(G_h^*,X_h)$ forms a solution for $(G,S)$.\qed 
\end{proof}

Lemma~\ref{lem:F-comp} gives us the following step.

\medskip
\noindent{\bf Step~6.}
If $I$ is disconnected, then solve
{\sc Induced Disjoint Paths} for each $(G_h^*,X_h)$. Return {\tt yes} if the return value is {\tt yes} for all of these instances and return {\tt no} otherwise.

\medskip
By Step~6, we may assume that $I$ is connected.
Then, by Lemma~\ref{lem:F}, we may assume that $I$ is a path.
This leads to the following new condition:

\medskip
\noindent
{\bf Condition iv)}
There is interference between two sets $S_i$ and $S_j$ for some $1\leq i<j\leq r$ if and only if $j=i+1$.

\medskip
\noindent
We are now ready to be a bit more precise, which is necessary for our algorithm. 
For $i=1,\ldots, r-1$, let $W_i$ be the set of all vertices $u\in V_G$, such that there is an edge 
$uv\in E_G$ with the following property: there are interfering terminal pairs $(s_p,t_p)$ and $(s_q,t_q)$ in $S_{i}$ and $S_{i+1}$, respectively, such that $G_p$ has an induced $(s_p,t_p)$-path containing $u$
and $G_q$ has an induced $(s_q,t_q)$-path containing $v$. Using this definition we can state the next lemma.

\begin{lemma}\label{lem:interf-Hi}
For $i=1,\ldots,r-1$, there is a set  $Z_i\subseteq T_i$ of at most two terminals such that $W_i\subseteq N(Z_i)$.
\end{lemma}

\begin{proof}
By Lemma~\ref{lem:interf}, each vertex from $W_i$ is in the neighbourhood of some terminal vertex in $T_i$. By  Lemma~\ref{lem:triangle}, $G[T]$ is triangle free. 
Hence, to obtain a contradiction, it is sufficient to suppose that there are two non-adjacent terminal vertices $z_1,z_2$ in $T_i$ and vertices $u_1\in (N(z_1)\setminus N(z_2))\cap W_i$, $u_2\in (N(z_2)\setminus N(z_1))\cap W_i$. For any pair of terminals $(s_j,t_j)$ from $S_{i}$ or $S_{i+1}$,  there is an $(s_j,t_j)$-path $R_j$ in 
$G_j$ by Condition~iii).
Moreover, by the definition of $W_i$, 
there are vertices $v_1$ and $v_2$ adjacent to $u_1$ and $u_2$, respectively, such that 
there are two terminal pairs $(s_p,t_p)$ and $(s_q,t_q)$ in $S_{i+1}$ and there are 
$(s_p,t_p)$ and $(s_q,t_q)$-paths $Q_p$ and $Q_q$ in $G_p$ and $G_q$ that pass through $v_1$ and~$v_2$, respectively.

Consider the subgraph~$R$ of $G$ induced by all the vertices of the paths $R_j$, $Q_p,Q_q$, and the vertices $u_1,u_2$. Let $w$ be an arbitrary terminal in $S_{i+1}$. It remains to observe that:
\begin{itemize}
\item the connectedness of $H_i$ implies that using the paths $R_j$ there is a path in $R$ from $z_1$ to $z_2$ that avoids the closed neighbourhood of any terminal vertex in $T_{i+1}$, and in particular, of $w$;
\item the fact that $z_2 \not\in T_{i+1}$ and $u_1 \not\in N(z_2)$ implies that using the vertices $u_1$, $v_1$, and the path $Q_p$ there is a path in $R$ from $z_1$ to $s_p$ that avoids the closed neighbourhood of $z_2$ and of any terminal in $T_{i+1} \setminus \{s_p\}$;
\item similarly, there is a path in $R$ from $z_2$ to $s_q$ that avoids the closed neighbourhood of $z_1$ and of any terminal in $T_{i+1} \setminus \{s_q\}$;
\item the connectedness of $H_{i+1}$ implies that using the paths $R_j$ there is a path in $R$ from $s_p$ to $w$ and from $s_q$ to $w$ that avoids the closed neighbourhood of any terminal vertex in $T_{i}$, and in particular, of $z_1$ and $z_2$.
\end{itemize}
Hence, $z_1,z_2,w$ form an asteroidal triple, a contradiction.\qed
\end{proof}
 
We obtain the following result on the sets $Z_1,\ldots,Z_{r-1}$ defined
in Lemma~\ref{lem:interf-Hi}.

\begin{lemma}\label{lem:Zi}
The sets $Z_1,\ldots,Z_{r-1}$ can be found in polynomial time.
\end{lemma}

\begin{proof}
Let $1\leq i\leq r-1$.
By Lemma~\ref{lem:interf} we can compute the set $W_i$ in polynomial time (see also the proof of Lemma~\ref{lem:interf-poly}).
By Lemma~\ref{lem:interf-Hi}, it suffices to determine as $Z_i$ 
a set of at most two of terminals in $T_i$ that are adjacent to every vertex of $W_i$.
This can be done in polynomial time as well. \qed 
\end{proof}

\subsection{Phase 3: Dynamic Programming}\label{sec:dp}
We are now ready to give a dynamic-programming algorithm for {\sc Induced Disjoint Paths}.
For simplicity, we solve the decision problem here, that is, we only check for the existence of paths, but the algorithm can be modified to get the paths themselves (if they exist).

Our algorithm is based on the information obtained from Phases~1 and~2, which leads to the following intuition. By Condition~iv), there is only interference between sets $S_i$ and $S_{i+1}$ for $1 \leq i < r$. By Lemma~\ref{lem:interf}, this interference is restricted to the neighbourhoods of the terminals of the pairs in $S_i$ and $S_{i+1}$. Even stronger, by the definition of $W_i$ and Lemma~\ref{lem:interf-Hi}, the interference caused by terminals in $T_i$ is restricted to the neighbourhoods of the terminal vertices in $Z_i$, where $|Z_i| \leq 2$. By Lemma~\ref{lem:degree}, each terminal vertex represents at most five terminals. Hence, there are at most ten vertices in $N(Z_i)$ that cause interference between $S_i$ and $S_{i+1}$. Unfortunately, we do not know these ten vertices. However, since $|Z_i| \leq 2$ and $Z_i$ can be found in polynomial time (by Lemma~\ref{lem:Zi}), we can enumerate all possibilities for these ten vertices in polynomial time. In fact, it suffices to enumerate only the set $Y$ of vertices among the ten vertices that are actually used by the paths in a solution.

This intuition leads to the following dynamic programming algorithm: for each $i\in\{1,\ldots,r\}$ and 
$Y\subseteq N(Z_{i})$ of size at most~$10$, 
we define a routine that solves {\sc Induced Disjoint Paths} for the graph~$F_i=G[V\setminus(\bigcup_{j\in\{i+1,\ldots r\}}\bigcup_{u\in T_j}N[u])]$ with the set of terminal pairs 
$S_i'=\bigcup_{j=1}^iS_j$ under the following 
additional condition: the set of non-terminal vertices from $N(Z_i)$ used by the paths in a solution is a subset of $Y$.
Call this routine {\tt Sol}$(i,Y)$. We execute this routine sequentially for $i=1,\ldots,r$, and we are clearly interested in deciding whether {\tt Sol}$(r,Y)$ returns {\tt yes} for some $Y \subseteq N(Z_r)$. The crux then is to show what the return value for {\tt Sol}$(i,Y)$ for a particular value of $i$ should be, using the return values for {\tt Sol}$(i-1,\cdot)$.

To this end, we construct the subroutine {\tt Component}$(i,X,Y)$: for each $i\in\{1,\ldots,r\}$
and for any two sets $X\subseteq N(Z_{i-1})$, $Y\subseteq N(Z_{i})$ of size at most~$10$ each, 
the subroutine solves {\sc Induced Disjoint Paths} 
for the graph~$F_i=G[V\setminus(\bigcup_{j\in\{1,\ldots r\},j\neq i}\bigcup_{u\in T_j}N[u])]$ with set of terminal pairs $S_i$ 
under the following two additional conditions:
\begin{itemize}
\item[a)] the paths from a solution are not adjacent to the vertices of $X$, and
\item[b)] the set of non-terminal vertices from $N(Z_i)$ used by the paths in a solution is a subset of $Y$.
\end{itemize}
Then, to determine the return value of {\tt Sol}$(i,Y)$, we consider all $X\subseteq N(Z_{i-1})$ for which {\tt Sol}$(i-1,X)$ returns {\tt yes} and let {\tt Sol}$(i,Y)$ return {\tt yes} if {\tt Component}$(i,X,Y)$ returns {\tt yes}. Otherwise, we let {\tt Sol}$(i,Y)$ return {\tt no}.

The correctness of the above algorithm is clear from the preceding discussion, the description of the dynamic program routine {\tt Sol}, and the description of the subroutine {\tt Component}. Moreover, the algorithm will run in polynomial time if the desired subroutine {\tt Component} can be constructed to run in polynomial time. We claim that {\tt Component} can indeed be so constructed, and the remainder of this section is dedicated to proving this claim.

\subsubsection{Constructing the Component Subroutine.} \label{sec:comp}
We are looking for an induced subgraph~$H'_i$ in $F_i$ obtained by subdivisions $P_l'$ of the paths $P_l$ of $H_i$ (under some additional constraints). We call this a \emph{solution}. Recall that $D_i$ is a dominating path of $H_i$ that joins terminals $x_i$ and $y_i$, which means that $D_i$ is a shortest path between $x_i$ and $y_i$ such that each vertex of $H_i$ lies on $D_i$ or has a neighbour on $D_i$. We construct $D_i$ by a breadth-first search that puts path vertices in its queue first (this is only done to improve the running time and does not influence correctness).

In order to find $H_i'$, we will attempt to ``trace'' the $(x_i,y_i)$-path $D_i'$ in $F_i$ that is in the solution $H_i'$ and is a subdivision of $D_i$. 
By Lemma~\ref{lem:Pi-deg}, we are interested in only a bounded number of vertices adjacent to the vertices of $D_i'$, and by 
Lemma~\ref{lem:subdiv}, the paths outside $D_i'$ 
have length at most~3
and~$D_i'$ dominates almost all vertices of these paths. Hence, vertices of $H_i'$ are ``close'' to $D_i'$ and their number is ``small''. The crux, therefore, is to find $D_i'$. To this end, we devise a dynamic-programming algorithm that maintains a small part of $D_i'$ and of the solution that we are building. We show that this is indeed possible by relying on the AT-freeness of $F_i$.

To formalize the preceding intuition, we require some definitions. Let $u_1,\ldots,u_p$ be the terminal vertices on $D_i$ enumerated in path order, that is, $x_i=u_1$ and $y_i=u_p$. 
Recall that $P_l$ is an $(s_l,t_l)$-path in $H_i$ for $(s_l,t_l)\in S_i$; each $P_l$ has a single internal vertex, its path vertex. 
For $j\in\{1,\ldots,p\}$, let $U_j$ be the set of terminal vertices in 
$(N_G(u_j)\setminus\{u_1,\ldots,u_p\})\cup\{u_j\}$.
Observe that $\bigcup_{j=1}^{p} U_j = T_i$, because $D_i$ is a dominating path of $H_i$ and if $D_i$ contains a path vertex (i.e.~a non-terminal vertex of $H_i$), then $D_i$ also contains both its neighbours, the vertices of the corresponding terminal pair (note that the ends of $D_i$ are terminal vertices).
From this observation and Condition~ii), for any pair  $(s_l,t_l) \in S_i$, we have that $s_l\in U_{j}$ and $t_l\in U_{j'}$ for some $j,j'\in\{1,\ldots,p\}$ and $j\neq j'$. 
To simplify our arguments, we assume 
without loss of generality
that if $s_l\in U_{j}$ and $t_l\in U_{j'}$, then $j<j'$.
For a set $R'$ of vertices, we let $U(R') = \bigcup_{j \in R'} U_j$, where the union is over all terminals in both $R'$ and $D_i$.
For a vertex~$z$ of $D_i$, we say that a terminal~$u\in T_i$ is \emph{behind} $z$ if $u\in U_j$ for $u_j$ in the $(x_i,z)$-subpath of $D_i$; otherwise we say that a terminal~$u\in T_i$ is \emph{ahead} of $z$.
By our earlier observation, the notions of behind and ahead of $z$ induce a partition of $T_i$.

Now, as mentioned before, we would like to ``trace'' an $(x_i,y_i)$-path $D_i'$ in $F_i$ that is a subdivision of $D_i$. We will argue that, in the dynamic program, it suffices to maintain only the last five vertices of the path $D_i'$ under construction. To this end, we require the following structural property.

\begin{lemma} \label{lem:struct-at}
Let $C = \{x_1,\ldots,x_t\}$ be the ordered set of vertices of a cycle in an AT-free graph 
for some $t\geq 3$. 
Then there exist integers $i,j$ with $1 \leq i < t$, $2 \leq j \leq 4$ 
and $i+j\leq t$ 
such that $x_i$ and $x_{i+j}$ are adjacent.\footnote{Note that the lemma does \emph{not} use indices module $t$. Although the lemma is certainly true in that setting, this is not how we will use it.}
\end{lemma}

\begin{proof}
Suppose $C$ has a chord. Then there exist indices $\alpha,\beta$ with $1 \leq \alpha < \alpha + 1 < \beta \leq t$ such that $x_{\alpha}x_{\beta}$ is an edge. Choose the chord such that $\beta-\alpha$ is smallest. If $\beta-\alpha \leq 4$, then the lemma holds. Hence, we may assume otherwise. Note that for any integers $\gamma,\delta$ with $\alpha \leq \gamma < \gamma + 1 < \delta \leq \beta$ it holds that either $v_{\gamma}v_{\delta}$ is not an edge (this would contradict the choice of $\alpha,\beta$) or $\gamma = \alpha$ and $\delta = \beta$. Hence, the vertices of $v_{\alpha},\ldots,v_{\beta}$ induce a cycle in the graph, which has length at least $6$ since $\beta-\alpha > 4$. Such a chordless cycle trivially has an asteroidal triple for the graph (for example, take $v_{\alpha}, v_{\alpha+2}, v_{\alpha+4}$), a contradiction.

It follows that the cycle has no chord. If $|C| \leq 5$, then the lemma follows immediately. Otherwise, the chordless cycle trivially contains an asteroidal triple for the graph (for example, take $v_1$, $v_3$, $v_5$), a contradiction. \qed
\end{proof}
We can now proceed with the description of the algorithm but first need a definition. Namely, a
pair of paths $P'_a$, $P'_b$ that are subdivisions of $P_a$, $P_b$ in $G_a$ and $G_b$, respectively, \emph{conflicts} if $P'_a$ intersects $P'_b$ in a vertex $v$ or there is an edge $vw \in E_G$ such that $v$ is an internal vertex of $P'_a$ and $w$ is an internal vertex of $P'_b$. In the former case, we call $v$ an \emph{intersection conflict vertex}, or simply a \emph{conflict vertex}. In the latter case, we call $vw$ a \emph{conflict edge} and both $v$ and $w$ \emph{conflict vertices}.

Note that it would appear another potential conflict between paths $P_a'$, $P_b'$ in $G_a$ and $G_b$ can occur, namely if there is an edge between an endpoint $v$ of (say) $P_a'$ and an internal vertex $w$ of $P_b'$. However, then $w \not\in V_{G_b}$ by definition (a contradiction) or $v \in \{s_b,t_b\}$ (which is allowed). Hence, if we are able to find paths $P'_l$ for all $(s_l,t_l) \in S_i$ such that $P'_l$ is an induced path in $G_l$ and the paths do not conflict, then these paths form a mutually induced set of paths.

Intuitively, ensuring there are no conflicts would be relatively straightforward to guarantee for pairs whose path vertices are not on $D_i$, because by combining that they should have length at most~$3$ by Lemma~\ref{lem:subdiv} and an argument that we need to consider few such pairs at any point in time, this limits the potential of conflict to polynomially many choices. However, for the other terminal pairs, this is much more complicated, because their paths $P'_l$ can have arbitrary length.

We need two more definitions.
Let $O_i \subseteq S_i$ be the set of pairs $(s_a,t_a) \in S_i$ for which the path vertex of $P_a$ is on $D_i$.
A \emph{walk} $W$ in $G$ is an ordered set $\{v_1,\ldots,v_{|W|}\}$ of vertices of $G$ such that $v_iv_{i+1} \in E_G$ for $i \in \{1,\ldots,|W|-1\}$.

Our algorithm is built on the intuition that we only ensure that $D_i'$ and the paths $P'_l$ with $l \in O_i$ are walks. In AT-free graphs, however, Lemma~\ref{lem:struct-at} will help prove that by considering the last five vertices of $D_i'$, we avoid all conflicts.

\medskip
\noindent 
{\bf Definition.}
{\it A \emph{realization} $(H_i',D_i', \{P'_l\})$ for $(H_i,D_i,\{P_l\})$ consists of:
\begin{itemize}
\item an induced subgraph $H'_i$ of $F_i$ such that no non-terminal vertices in $H'_i$ are adjacent to $X$ and all non-terminal vertices in $V_{H'_i} \cap N(Z_i)$ are in $Y$;
\item a walk $D'_i$ in $H_i'$ that is a subdivision of $D_i$; in particular, it visits $u_1,\ldots,u_p$ in order (without duplication);
\item for each $(s_l,t_l) \in S_i\setminus O_i$, there is an induced path $P'_l$ in $H'_i$ and in $G_l$ from $s_l$ to $t_l$ of length $2$ or $3$;
\item for each $(s_l,t_l) \in O_i$, there is a walk $P'_l$ in $H'_i$ and $G_l$ of length at least~$2$ from $s_l$ to $t_l$ and $P'_l$ is a subwalk of $D'_i$.
\end{itemize}
A realization is an \emph{embedding} if the set $\{P'_l\}$ of paths is a set of mutually induced disjoint paths.}

\medskip
\noindent
Note that if $(H_i',D_i', \{P'_l\})$ is an embedding for $(H_i,D_i,\{P_l\})$, then the paths $\{P_l'\}$ form a solution. Moreover, the conditions on the paths imply that $D_i'$ is an induced path of $F_i$. Indeed, then the only possible edges between vertices of $D_i'$ would be between non-consecutive terminals on $D_i'$. However, these edges are part of $G$ and thus $H_i$ as well. As $D_i'$ is a subdivision of $D_i$ by the definition of a realization, this contradicts that $D_i$ is a shortest path in $H_i$.

Observe that the preceding argument also implies that if $D_i'$ is not an induced path, then this is due to a conflicting pair or because $D_i'$ is not a subdivision of $D_i$.

Clearly, our goal is to find an embedding, which is stronger than a realization. As mentioned, our algorithm aims to find a realization. The crux is how to ensure that we end up with an embedding. The following lemma is crucial.

\begin{lemma} \label{lem:component-correct}
Suppose $(H_i',D_i', \{P'_l\})$ is a realization for $(H_i,D_i,\{P_l\})$ in an AT-free graph, but not an embedding. Then there is a subwalk $R'$ of $D'_i$ with $|R'| \leq 5$ such that:
\begin{itemize}
\item $R'$ is not an induced path, \emph{or}
\item $U(R')$ contains at least one of $\{s_a,t_a\}$ and at least one of $\{s_b,t_b\}$ for some conflict pair $P'_a$,$P'_b$; moreover, if $(s_a,t_a) \in O_i$ (or $(s_b,t_b) \in O_i$), then $R'$ contains a conflict vertex of $P'_a$ (or $P'_b$).
\end{itemize}
\end{lemma}
\begin{proof}
We start with the following claims.

\begin{quote}
(C1) $D'_i$ is a path, or a subwalk $R'$ as in the lemma statement exists.
\end{quote}
Suppose that $D'_i$ is not a path. Let $x$ be any vertex on the walk $D_i'$ such that the distance between two occurrences of $x$ on $D_i'$ is minimum. Let $x = x_1,\ldots,x_t=x$ be the subwalk of $D_i'$ between those two occurrences; that is, $x \not\in \{x_2,\ldots,x_{t-1}\}$, $x_2,\ldots,x_{t-1}$ are distinct, and $t$ is minimum over all choices of $x$. (Note that any $x_i$ might still occur multiple times on $D_i'$ outside of the subwalk.) It follows that the subwalk forms a cycle. Then by Lemma~\ref{lem:struct-at}, there exist integers $i,j$ with $1 \leq i \leq t-1$ and $2 \leq j \leq 4$ such that $x_{i} x_{i+j}$ is an edge. Then let $R' = \{x_i,\ldots,x_{i+j}\}$. This subwalk of $D'_i$ is not an induced path and has $|R'| \leq 5$.
\medskip

\begin{quote}
(C2) $D'_i$ is an induced path, or a subwalk $R'$ as in the lemma statement exists.
\end{quote}
By (C1), we may assume that $D'_i$ is a path. Let $x,y$ be two non-consecutive vertices of $D'_i$ such that $xy$ is an edge. Consider the cycle $x = x_1,\ldots,x_t = y$ formed by the edge $xy$ and the part of the path $D'_i$ between $x$ and $y$. Then by Lemma~\ref{lem:struct-at}, there exist integers $i,j$ with $1 \leq i \leq t-1$ and $2 \leq j \leq 4$ such that $x_{i} x_{i+j}$ is an edge. Then let $R' = \{x_i,\ldots,x_{i+j}\}$. This subwalk of $D'_i$ is not an induced path and has $|R'| \leq 5$.
\medskip

\begin{quote}
(C3) All paths in $\{P'_l\}$ are induced paths, or a subwalk $R'$ as in the lemma statement exists.
\end{quote}
Note that a path $P'_l$ with $(s_l,t_l) \in S_i \setminus O_i$ is an induced path by the definition of a realization. Since a path $P'_l$ with $(s_l,t_l) \in O_i$ is a subwalk of $D_i'$, the claim now follows from (C2).
\medskip

By the preceding claims, we may assume that $D_i'$ is an induced path and that all paths in $\{P'_l\}$ are induced paths. Hence, what is preventing the realization from being an embedding are conflicts. We now consider any conflicting pair and show that a subwalk $R'$ as in the lemma statement exists.

Suppose that $P'_a$ and $P'_b$ with $(s_a,t_a),(s_b,t_b) \in O_i$ conflict. Since both are subwalks of $D_i'$ by the definition of a realization, this implies that $D_i'$ is either not a path or not induced. Then we obtain a contradiction, because (C1) or (C2) implies that a subwalk $R'$ as in the lemma statement exists. Hence, we may now assume that no such conflicting pairs exists.

Suppose that $P'_a$ and $P_b'$ with $(s_a,t_a) \in O_i$ and $(s_b,t_b) \in S_i\setminus O_i$ conflict and have an intersection conflict vertex $v$. Since $P'_b$ has length at most~$3$, $v$ is a neighbour of $s_b$ or $t_b$; say $v \in N(s_b)$. Since $P'_a$ is a path in $G_a$, but contains $v$, it follows that $s_a = s_b$ or $t_a = s_b$; say $s_a = s_b$. Then $v \in N(s_a)$ and thus $v$ and $s_a$ are consecutive vertices of $P'_a$ and thus $D_i'$. Let $R' = \{v,s_a\}$. Note that $U(R')$ contains $s_a$ and $s_b$ and $R'$ contains a conflict vertex of $P'_a$. Thus, a subwalk $R'$ as in the lemma statement exists. Hence, we may now assume that no such conflicting pair exists. Since $D'_i$ is a subdivision of $D_i$ by the definition of a realization, this also implies that we may assume that no path~$P'_b$ with $(s_b,t_b) \in S_i\setminus O_i$ has an internal vertex that is a vertex of $D_i'$.

Suppose that $P'_a$ and $P_b'$ with $(s_a,t_a) \in O_i$ and $(s_b,t_b) \in S_i\setminus O_i$ conflict and have a conflict edge $vw$ but no intersection conflict vertex. Assume $v \in P'_a$ and $w \in P'_b$. By Lemma~\ref{lem:Pi-term}, without loss of generality $s_b \in D_i$ and thus $s_b \in D_i'$. Consider the cycle $C_{vw}$ that contains the edge $vw$ and the vertex $s_b$ of the graph induced by $P'_b$ and $D'_i$. Note that such a cycle indeed exists and is uniquely defined, and consists of the edge $vw$, a part of $P'_b$ and a part of $D'_i$. Now choose $a,b$ and $vw$ such that the length of $C_{vw}$ is minimum.

If $|C_{vw}| > 5$, then by Lemma~\ref{lem:struct-at}, $C_{vw}$ has a chord $xy$. This chord cannot have both endpoints in $D_i'$ or both in $P_b'$, or (C2) or (C3) implies that a subwalk $R'$ as in the lemma statement exists. In particular, combined with the fact that $P'_a$ and $P'_b$ are paths in $G_a$ and $G_b$ respectively by the definition of a realization, this implies that the set of endpoints of the chord is disjoint from $\{s_a,t_a,s_b,t_b\}$. Hence, the chord $xy$ is a conflict edge of $P'_a$ and $P'_b$ and the cycle $C_{xy}$ induced by it is shorter than $C_{vw}$. This contradicts the choice of $vw$.

It follows that $|C_{vw}| \leq 5$. Let $R'$ be the subpath of $C_{vw}$ between $s_b$ and $w$. Note that this is a subpath of $D'_i$. Moreover, $|R'| \leq 4$, because $|C_{vw}| \leq 5$ and $v$ is not on $R'$. Since $s_b \in R'$, $s_b \in U(R')$. Moreover, since $P'_a$ is a subpath of $D'_i$ and $s_b$ is a terminal on $R'$, it must be that $s_a$ or $t_a$ is in $R'$, and thus at least one of them is in $U(R')$. Finally, we note that $w \in R'$ by definition. Thus, a subwalk $R'$ as in the lemma statement exists. Hence, we may now assume that no such conflicting pair exists.

Suppose that $P'_a$ and $P_b'$ with $(s_a,t_a) \in S_i\setminus O_i$ and $(s_b,t_b) \in S_i\setminus O_i$ conflict and have an intersection conflict vertex $v$. Since $P'_b$ has length at most~$3$, $v$ is a neighbour of $s_b$ or $t_b$; say $v \in N(s_b)$. Since $P'_a$ is a path in $G_a$, but contains $v$, it follows that $s_a = s_b$ or $t_a = s_b$; say $s_a = s_b$. If $s_a$ (and thus $s_b$) is on $D_i$, then let $R' = \{s_a\}$ and we obtain a subwalk $R'$ as in the lemma statement. Otherwise, there is a terminal vertex~$u$ in $D_i$ such that $s_a \in N(u)$. Note that $u \in D'_i$. Let $R' = \{u\}$; then $s_a \in U(R')$ and we obtain a subwalk $R'$ as in the lemma statement. Hence, we may now assume that no such conflicting pair exists.

Suppose that $P'_a$ and $P_b'$ with $(s_a,t_a) \in S_i\setminus O_i$ and $(s_b,t_b) \in S_i\setminus O_i$ conflict and have a conflict edge $vw$ but no intersection conflict vertex. Assume $v \in P'_a$ and $w \in P'_b$. By Lemma~\ref{lem:Pi-term}, we may assume without loss of generality that $s_a \in D_i$ and $s_b \in D_i$. Hence, $s_a \in D_i'$ and $s_b \in D_i'$. Consider the cycle $C_{vw}$ formed by the edge $vw$, the subpath of $P'_a$ between $v$ and $s_a$, the subpath of $D'_i$ between $s_a$ and $s_b$ and the subpath of $P'_b$ between $s_b$ and~$w$. Choose $a,b$ and $vw$ such that the length of $C_{vw}$ is minimum.

If $|C_{vw}| > 5$, then by Lemma~\ref{lem:struct-at}, $C_{vw}$ has a chord $xy$. This chord cannot have both endpoints in $D'_i$, both in $P'_a$, or both in $P'_b$, or (C2) or (C3) implies that a subwalk~$R'$ as in the lemma statement exists. Moreover, by our earlier assumption, the chord cannot have an endpoint in $D'_i$ and an endpoint in one of $P'_a$ or $P'_b$. Also, combined with the fact that $P'_a$ and $P'_b$ are paths in $G_a$ and $G_b$ respectively by the definition of a realization, this implies that the set of endpoints of the chord is disjoint from $\{s_a,t_a,s_b,t_b\}$. Hence, the chord $xy$ is a conflict edge of $P'_a$ and $P'_b$ and the cycle $C_{xy}$ induced by it is shorter than $C_{vw}$. This contradicts the choice of $vw$.

It follows that $|C_{vw}| \leq 5$. Let $R'$ be the subpath of $C_{vw}$ between $s_a$ and $s_b$. Note that this is a subpath of $D'_i$. Moreover, $|R'| \leq 3$, because $|C_{vw}| \leq 5$ and $v$ and $w$ are not on $R'$ (as $v$ and $w$ are internal vertices of pairs $(s_a,t_a) \in S_i\setminus O_i$ and $(s_b,t_b) \in S_i\setminus O_i$). Since $s_a,s_b \in R'$, $s_a,s_b \in U(R')$. Thus, a subwalk $R'$ as in the lemma statement exists.
This completes our proof.
\qed\end{proof}

\subsubsection{Dynamic Program.}
We now define the table used in the subroutine {\tt Component}$(i,X,Y)$. Each entry uses the following fields:
\begin{itemize}
\item an integer $\ell \geq 0$;
\item an ordered set of vertices $R' \subseteq V_{F_i}$ of at most five vertices of $F_i$ such that $F_i[R']$ is a path;
\item a vertex $z \in D_i$;
\item a set of non-terminal vertices $N' \subseteq V_{F_i}$;
\end{itemize}
and stores whether there exists a subgraph $H_i'$ of $F_i$ with the following properties:
\begin{itemize}
\item no non-terminal vertices in $H'_i$ are adjacent to $X$ and all non-terminal vertices in $V_{H'_i} \cap N_{F_i}(Z_i)$ are in $Y$;
\item there is a walk $D'_i$ in $H_i'$ of length $\ell$ that is a subdivision of the subpath of $D_i$ from $x_i$ to $z$;
\item for each $(s_l,t_l) \in S_i\setminus O_i$ for which at least one terminal is behind $z$, there is an induced path $P'_l$ in $H'_i$ and $G_l$ from $s_l$ to $t_l$ of length $2$ or $3$; if $s_l$ or $t_l$ is in $U(R')$, then all internal vertices of this path are in $N'$;
\item for each $(s_l,t_l) \in O_i$ for which both terminals are behind $z$, there is a walk $P'_l$ in $H'_i$ and $G_l$ of length at least~$2$ from $s_l$ to $t_l$ and $P'_l$ is a subwalk of $D'_i$;
\item if a single terminal of $(s_l,t_l) \in O_i$ is behind $z$ and $z$ is a path vertex, there is a walk $P'_l$ in $H'_i$ and $G_l$ of length at least~$2$, starting in $s_l$ but not ending in $t_l$, and $P'_l$ is a subwalk of $D'_i$;
\item the path $F_i[R']$ and all paths for terminal pairs $(s_l,t_l) \in S_i\setminus O_i$ for which $s_l$ or $t_l$ is in $U(R')$ are mutually induced disjoint;
\item $R'$ is equal to the last $\min\{5,|D'_i|\}$ vertices of $D_i'$;
\item $N' \subseteq V_{H_i'}$;
\item there is a subpath $R$ of $D_i$, ending in $z$, and a map from $R'$ to $R$ such that any terminal vertex $v$ on $R'$ is mapped to the same vertex on $R$ and any non-terminal vertex $v$ on $R'$ is mapped to a path vertex on $R$. Moreover, the map respects the order of the terminal and path vertices of $R$.
\end{itemize}
The last condition is called the \emph{mappability condition} and is fully determined by $R'$ and $z$. We can always pick the subpath $R$ to have length at most the length of $R'$, and thus at most~$5$. The first five conditions effectively ensure that $H'_i$ is a partial embedding of all terminal pairs with a terminal vertex behind $z$.

We use $\mathrm{last}(R)$ respectively $\mathrm{last}(R')$ to denote the last vertex of $R$ respectively $R'$. We also use $\mathrm{term}(z)$ to be equal to $(s_l,t_l) \in O_i$ for which either $z$ is the path vertex or the vertex immediately following $z$ on $D_i$ is the path vertex; otherwise, $\mathrm{term}(z)$ is undefined.

\paragraph{Table construction.}
The table is initialized by having an entry for $\ell = 0$, $z = x_i$, $R' = \{x_i\}$, and each set of non-terminal vertices $N'$ for which $F_i[N' \cup \{x_i\}]$ satisfies the conditions of an entry. In particular, for each $(s_l,t_l) \in S_i\setminus O_i$ for which $\{s_l,t_l\} \cap U(R') \not= \emptyset$, there is an induced path $P'_l$ in $H'_i$ and $G_l$ from $s_l$ to $t_l$ of length $2$ or $3$ for which all internal (non-terminal) vertices are in $N'$, and all such paths are mutually induced disjoint. Also, no non-terminal vertices in $N'$ are adjacent to $X$ and all non-terminal vertices in $N' \cap N_{F_i}(Z_i)$ are in $Y$.

Suppose the table has been fully filled for $\ell-1$. Consider all vertices $u \in V_{F_i}$ for which the table contains an entry for $(\ell-1,R',z,N')$ such that $u$ is adjacent to $\mathrm{last}(R')$ and is not adjacent to any other vertices in $R'$.
For each such $u$ and each entry $(\ell-1,R',z,N')$ with these properties, we create (possibly) new entries for the table for $\ell$. We call these new entries \emph{successor entries}. We distinguish two cases.

\medskip
\noindent
{\bf Case 1.} $u \not\in T_i$.\\ Note that in this case, $u$ will be an internal vertex of a path in the solution for a terminal pair in $O_i$. Hence, it cannot be adjacent to internal vertices of other paths and terminals except, possibly, the terminals joined by the path. We only construct new entries if the following conditions hold:
\begin{itemize}
\item $u$ is not adjacent to any vertex in $N'$;
\item $u$ is not adjacent to $X$ and if it is in $N_{F_i}(Z_i)$, then it is in $Y$;
\item $z$ is a path vertex or the vertex on $D_i$ after $z$ is a path vertex;
\item $u$ is a vertex in $G_l$ for the terminal pair $(s_l,t_l) = \mathrm{term}(z)$.
\end{itemize}
If the above holds, then entries $(\ell,R'_n,z_n,N'_n)$ are constructed as follows.  If $z$ is a path vertex, then $z_n=z$; otherwise, $z_n$ is the vertex on $D_i$ after $z$.

If $|R'| \leq 4$, then $R'_n$ is equal to $R'$ and $u$ is added as the last vertex of the ordered set. Let $R$ and $R_n$ be the subpaths of $D_i$ that follow from the mappability condition on $R'$ and $z$ respectively $R'_n$ and $z_n$. Note that $R_n \supseteq R$ and the set of terminal vertices of $R_n$ and $R$ are equal. Hence, we can set $N'_n = N'$ by the above conditions.

If $|R'| = 5$, let $x$ be the first vertex of $R'$. Then $R'_n$ is equal to the last four vertices of $R'$ and $u$ is added as the last vertex of the ordered set. 
Let $R$ and $R_n$ be the subpaths of $D_i$ that follow from the mappability condition on $R'$ and $z$ respectively $R'_n$ and $z_n$. If $x$ is a non-terminal vertex, then the sets of terminal vertices of $R_n$ and $R$ are equal and we can set $N'_n = N'$. Otherwise, obtain $N'_n$ from $N'$ by removing any vertices from $N'$ belonging to a path for a terminal $(s_l,t_l) \in S_i \setminus O_i$ for which neither $s_l$ nor $t_l$ is in $U(R'_n)$. In either case, $N'_n$ is correctly set by the above conditions.

\medskip
\noindent
{\bf Case 2.} $u \in T_i$.\\ In this case, $u$ can be adjacent to other terminals and $u$ can be a terminal vertex for several terminal pairs. This makes the construction of entries slightly more involved. We only construct new entries if $u$ is the vertex on $D_i$ after $z$. In that case, entries $(\ell,R'_n,z_n,N'_n)$ are constructed as follows. Set $z_n = u$.

If $|R'| \leq 4$, then $R'_n$ is equal to $R'$ and $u$ is added as the last vertex of the ordered set. If $|R'| = 5$, then $R'_n$ is equal to the last four vertices of $R'$ and $u$ is added as the last vertex of the ordered set. Let $R$ and $R_n$ be the subpath of $D_i$ that follow from the mappability condition on $R'$ and $z$ respectively $R'_n$ and $z_n$.

Now construct $N'_n$ from $N'$. If $|R'| = 5$, let $x$ be the first vertex of $R'$. If $x$ is a terminal vertex, then remove any vertices from $N'$ belonging to a path for a terminal $(s_l,t_l) \in S_i \setminus O_i$ for which neither $s_l$ nor $t_l$ is in $U(R'_n)$. 

We might also add several possible sets of paths to $N'_n$. For each terminal $(s_l,t_l) \in S_i \setminus O_i$ for which $s_l$ or $t_l$ is in $U(R'_n) \setminus U(R')$ while neither are in $U(R')$, consider all possible induced paths in $G_l$ between $s_l$ and $t_l$ of which the internal vertices are not adjacent to any vertex in $R'$, $R'_n$, and $N'$. Consider all possible combinations of such paths for such terminal pairs, and add the internal vertices of them to $N'_n$. Note that this creates many different sets $N'_n$, each of which yields a new table entry.

\medskip
\noindent
{\bf Decision and analysis.}
We finish the description of the subroutine {\tt Component}$(i,X,Y)$. It can be readily verified that the entries constructed as above satisfy all stated conditions. Finally, if there is an entry in the table for $(\ell,R',y_i,N')$ for any $\ell,R',N'$, then return {\tt yes}; otherwise, return {\tt no}.

The crucial property of {\tt Component}$(i,X,Y)$ is given in the following lemma.

\begin{lemma}\label{lem:comp-iXY}
For each $i\in\{1,\ldots,k\}$, the subroutine {\tt Component}$(i,X,Y)$ tests in polynomial time whether $(H_i,D_i,\{P_l\})$ has an embedding.
\end{lemma}
\begin{proof}
Suppose that $(H_i,D_i,\{P_l\})$ has an embedding $(H'_i,D'_i,\{P_l'\})$. Let $v_1,\ldots,v_j$ be vertices of $D'_i$ in order from $x_i=v_1$ to $y_i=v_j$. By induction, it can be readily shown from the construction of the entries that there is a sequence of successor entries $(\ell,\{v_{\ell-4},\ldots,v_\ell\},z_\ell,N'_\ell)$ for $1 \leq \ell \leq j$ and appropriate $z_\ell,N'_\ell$ that traces along $D_i'$ for which there is an entry for $(j,\{v_{j-4},\ldots,v_j\},y_i,N'_j)$ for some appropriate $N'_j$. Hence, the subroutine will return {\tt yes}.

Suppose that the subroutine returns {\tt yes}. Let $(\ell,R'_\ell,z_\ell,N'_\ell)$ be a sequence of successor entries for $1 \leq \ell \leq j$, $R'_\ell, z_\ell,N'_\ell$ for which there is an entry for $(j,R'_j,y_i,N'_j)$ for some $R'_j,N'_j$. Let $D'_i$ be the walk induced by the sets $R'_\ell$. Let $H'_i$ be the subgraph of $F_i$ induced by the terminal vertices $T_i$, the walk $D_i'$, and $\bigcup_{1 \leq \ell \leq j} N'_\ell$. Let $P'_l$ be the set of walks for terminal pairs $(s_l,t_l)$ implied by the entries $(\ell,R'_\ell,z_\ell,N'_\ell)$.

Note that the walks for terminal pairs $(s_l,t_l) \in O_i$ are uniquely defined by $D'_i$. Then by Lemma~\ref{lem:component-correct}, claims (C1) and (C2), and the fact that each set $R'_\ell$ is an induced path by definition, it follows that $D_i'$ is an induced path.

For terminal pairs $(s_l,t_l) \in S_i \setminus O_i$, a small argument is needed, because potentially there is no $R'_\ell$ for which $s_l,t_l \in U(R'_\ell)$, meaning that a path for $(s_l,t_l)$ appears, disappears, and then re-appears and is potentially different. Let $\ell_l$ be smallest such that $\{s_l,t_l\} \cap U(R'_{\ell_l})$. Without loss of generality, $s_l \in U(R'_{\ell_l})$. Suppose that $t_l \not\in U(R'_{\ell_l}) \cup \cdots \cup U(R'_{\ell_l+4})$. Note there is an induced path $P'_l$ in $G_l$ between $s_l$ and $t_l$ stored in $N'_{\ell_l},\ldots,N'_{\ell_l+4}$ that is the same path for each entry. Moreover, by construction, no vertices of this path are adjacent to vertices of $R'_{\ell_l},\ldots,R'_{\ell_l+4}$. Let $v$ be the first vertex on $D'_i$ after $\mathrm{last}(R'_{\ell_l+4})$ that is adjacent to a vertex of $P'_l$. Since $t_l$ is not in $U(R'_{1}),\ldots,U(R'_{\ell_l+4})$ by definition and $t_l \in U(D'_i)$ by definition, $v$ indeed exists. Then the subpath of $D'_i$ from $s_l$ to $v$ plus the subpath of $P'_l$ from $s_l$ to the neighbour of $v$ on $P'_l$ induce a cycle in $F_i$. Since the subpath of $D'_i$ from $s_l$ to $v$ contains at least six vertices, the length of the cycle is at least~$6$, a contradiction to the AT-freeness of $F_i$. Hence, the sequence of entries uniquely defines a path $P'_l$ between $s_l$ and $t_l$ for each terminal pair $(s_l,t_l) \in S_i \setminus O_i$.

Consider $(H'_i,D'_i,\{P_l'\})$. By construction, we know that $(H'_i,D'_i,\{P_l'\})$ is a realization of $(H_i,D_i,\{P_l\})$. For sake of contradiction, suppose that $(H'_i,D'_i,\{P_l'\})$ is not an embedding of $(H_i,D_i,\{P_l\})$. Then by Lemma~\ref{lem:component-correct}, there is a subwalk $R'$ of $D'_i$ with $|R'| \leq 5$ such that:
\begin{itemize}
\item $R'$ is not an induced path, \emph{or}
\item $U(R')$ contains at least one of $\{s_a,t_a\}$ and at least one of $\{s_b,t_b\}$ for some conflict pair $P'_a$,$P'_b$; moreover, if $(s_a,t_a) \in O_i$ (or $(s_b,t_b) \in O_i$), then $R'$ contains a conflict vertex of $P'_a$ (or $P'_b$).
\end{itemize}
However, the first case is excluded by the fact that each set $R'_\ell$ is an induced path by definition. The second case is also excluded. Let $(s_a,t_a) \in S_i$ and $(s_b,t_b) \in S_i$ be terminal pairs for which $U(R'_\ell)$ contains at least one of $\{s_a,t_a\}$ and at least one of $\{s_b,t_b\}$.
\begin{itemize}
\item If $(s_a,t_a) \in O_i$ and $(s_b,t_b) \in O_i$, then their paths cannot conflict, because $R'_\ell$ would contain conflict vertices of both pairs, but $R'_\ell$ is an induced path.
\item If $(s_a,t_a) \in O_i$ and $(s_b,t_b) \in S_i \setminus O_i$, then their paths cannot conflict, because $R'_\ell$ would contain a conflict vertex and $P'_b$ would be contained in $N'_\ell$ and thus is not be adjacent or intersect with $R'_\ell$ by construction.
\item If $(s_a,t_a) \in S_i \setminus O_i$ and $(s_b,t_b) \in S_i \setminus O_i$, then their paths cannot conflict by the construction of $N'_\ell$, because at least one of $\{s_a,t_a\}$ and at least one of $\{s_b,t_b\}$ and thus both $P'_a$ and $P'_b$ are contained in $N'_\ell$.
\end{itemize}
It follows that $(H'_i,D'_i,\{P_l'\})$ is an embedding of $(H_i,D_i,\{P_l\})$.

To prove polynomial time complexity, it suffices to bound $|N'|$. We argue that if $(H_i,D_i,\{P_l\})$ has an embedding $(H'_i,D'_i,\{P_l'\})$ and $(\ell,R'_\ell = \{v_{\ell-4},\ldots,v_\ell\},z_\ell,N'_\ell)$ for $1 \leq \ell \leq |D_i'|$ is the corresponding sequence of successor entries, then $|N'_\ell|$ is bounded by~47 for each $\ell$. Hence, we may limit the size of $N'$ to~47 during our dynamic program. By Lemma~\ref{lem:subdiv}, $D_i'$ will dominate all but at most two vertices of $H'_i$. We call these the special vertices. As argued in Lemma~\ref{lem:subdiv}, the special vertices are internal vertices of some $(s_l,t_l)$ paths $P'_j$.

We claim that all non-terminals in $N'_\ell$ that are not special are adjacent to a terminal in $R'$ or the two terminal vertices preceding $R'$ on $D_i'$ or the two terminal vertices succeeding $R'$ on $D_i'$. Call the latter two sets of terminal vertices $K$ and $L$ respectively. We consider two cases.

In the first case, consider a terminal pair $(s_l,t_l)$ for which $t_l \in U(R'_\ell) \setminus D_i'$ (note that this implies that $(s_l,t_l) \in S_i\setminus O_i$). Then $s_l \in D_i'$ by Lemma~\ref{lem:Pi-term}. By Lemma~\ref{lem:subdiv}, $P_l'$ has length at most~$3$. If $P_l'$ has length~$2$, then the single non-terminal vertex on $P_l'$ is adjacent to $s_l$. If $P_l'$ has length~$3$, then the vertex of $P_j'$ that follows $t_l$ is not dominated by $D_i'$, or the set of paths is not mutually induced. Note that this vertex is special. Moreover, the other internal vertex of $P_j'$ is adjacent to $s_l$. Now let $v \in R'_\ell$ such that $t_l$ is a neighbour of $v$ in $H_i'$. Suppose that $s_l$ precedes the vertices of $K$ on $D_i'$ (the case that $s_l$ succeeds $L$ on $D_i'$ is similar). 
Then, as $D_i'$ is a subdivision of $D_i$, the subpath of $D_i$ between $s_l$ and $v$ has length at least~$3$, whereas the path in $H_i$ between $s_l$ and $v$ via the path vertex of $P_l$ and $t_l$ has length at most~$3$. Then either we contradict that $D_i$ is a shortest path or, since $D_i$ was constructed by a breadth-first search that puts path vertices in its queue first, the former path would be preferred. Hence, $s_l \in K \cup R'_\ell \cup L$, implying the claim in the first case.

In the second case, consider a terminal pair $(s_l,t_l)$ for which $t_l \in R'_\ell$ and $(s_l,t_l) \in S_i\setminus O_i$. If $s_l \not\in D_i'$, then following the preceding argument, any internal vertices of $P_l'$ that are not adjacent to $t_l$ are not on $D_i'$, and thus special. Otherwise, $s_l \in D_i'$. By Lemma~\ref{lem:subdiv}, $P_l'$ has length at most~$3$, and thus its internal vertices are adjacent to $s_l$ or $t_l$. Suppose that $s_l$ precedes the vertices of $K$ on $D_i'$ (the case that $s_l$ succeeds $L$ on $D_i'$ is similar). Again, the subpath of $D_i$ between $s_l$ and $t_l$ has length at least~$3$, whereas the path in $H_i$ between $s_l$ and $t_l$ via the path vertex of $P_l$ has length at most~$2$, contradicting that $D_i$ is a shortest path. Hence, $s_l \in K \cup R'_\ell \cup L$, implying the claim in the second case and overall.

Now following Lemma~\ref{lem:Pi-deg}, any terminal vertex of $D_i$ is adjacent to at most five path-vertices of $H_i$. This implies that any terminal vertex of $D_i'$ is adjacent to at most five non-terminal vertices of $H_i'$. Following the claim, all non-terminals in $N'_\ell$ (except the special vertices) are adjacent to a terminal in $K \cup R' \cup L$. By definition, there are at most~$9$ such vertices, implying that $|N'_\ell| \leq 9\cdot5+2 = 47$.
It follows that the size of the table is polynomially bounded. It can then be readily seen that the algorithm runs in polynomial time.
\qed\end{proof}

The full algorithm, its correctness, and its running time now follows from Lemma~\ref{lem:comp-iXY} and the preceding discussion.

\begin{theorem}\label{thm:main}
The {\sc Induced Disjoint Paths} problem can be solved in polynomial time for AT-free graphs.
\end{theorem}

\section{Induced Topological Minors}\label{s-itm}

Recall that a graph~$G$ contains a graph~$H$ as an induced topological minor if $G$ contains an induced subgraph isomorphic to a subdivision of $H$, that is, to a graph obtained from $H$ by a number of edge subdivisions, and that the {\sc Induced Topological Minor} problem is the corresponding decision problem. 
Recall also that a graph~$G$ contains a graph~$H$ on vertices $x_1,\ldots,x_k$ as an anchored induced topological minor if there exist $k$ pre-specified vertices $u_1,\ldots,u_k$ in $G$, such that $G$ contains an induced subgraph~$H'$ isomorphic to a subdivision of $H$ and the isomorphism maps $x_i$ to $u_i$
 for $i=1,\ldots,k$. In that case we say that $H$ is {\it anchored in $V_H\subseteq V_G$}. The {\sc Anchored Induced Topological Minor} is the corresponding decision problem.

As mentioned in Section~\ref{s-intro} one can use an algorithm that solves {\sc Induced Disjoint Paths} in polynomial time for any fixed integer~$k$ on some graph class ${\cal G}$ as a subroutine for obtaining an algorithm that solves {\sc $H$-Induced Topological Minor} in polynomial time for any fixed graph~$H$ on ${\cal G}$. This is because
{\sc Anchored Induced Topological Minor} can be reduced to {\sc Induced Disjoint Paths}. This reduction is known, see for example Corollary 3 of the paper of Belmonte et al.~\cite{BGHHKP12}, where ${\cal G}$ is the class of chordal graphs.
Due to Theorem~\ref{thm:main}, we can now also let ${\cal G}$ be the class of AT-free graphs. Below we provide a short proof for this;. In particular we do this in order to show that the reduction works when $H$ has isolated vertices. This case is not yet captured as condition~i) excludes the case when $s_i=t_i$ for some $1\leq i\leq k$.

\begin{corollary}\label{cor:anchor}
The {\sc Anchored Induced Topological Minor} problem can be solved in polynomial time for AT-free graphs.
\end{corollary}

\begin{proof}
Let $G$ and $H$ be graphs, such that $H$ is anchored in $V_H\subseteq V_G$. 
If $H$ has no isolated vertices, then 
we reduce {\sc Anchored Induced Topological Minor} to {\sc Induced Disjoint Paths} on $G$ by constructing 
the pair of terminals $(u,v)$ for each edge $uv\in E_H$. 

Suppose now that $H$ has at least one isolated vertex. Then for each isolated vertex~$u$ in $H$, we do the following. First, we check whether $u$ is adjacent in $G$ to some other vertex of $H$. If so, then we stop and return {\tt no}.
Otherwise, we remove $u$ from $H$ and remove $u$ with its neighbourhood in $G$ from the graph~$G$.
If we
removed all vertices of~$H$, then we return {\tt yes}. Otherwise,
denote by $G'$ and $H'$ the graphs obtained from $G$ and $H$, respectively, by these removals. 
Now for each edge $uv\in E_{H'}$, we construct the pair of terminals $(u,v)$ in $G'$ and solve the obtained instance of  
{\sc Induced Disjoint Paths} on $G'$. 
\qed
\end{proof}

If $H$ is a \emph{fixed} graph, that is, not part of the input, then {\sc Induced Topological Minor} can be solved in polynomial time
by a standard reduction, as mentioned in Section~\ref{s-intro}.
To make our paper self-contained, we give a short proof.

\begin{corollary}\label{cor:ITP-XP}
The {\sc Induced Topological Minor} problem can be solved in time $n^{k+O(1)}$ for pairs $(G,H)$ where $G$ is an $n$-vertex AT-free graph and $H$ is a $k$-vertex graph.
\end{corollary}

\begin{proof}
Let $G$ be an AT-free graph on $n$ vertices, and let $H$ be a graph with $V_H=\{x_1,\ldots,x_k\}$ for some $k\leq n$. 
First suppose that $H$ has no isolated vertices.
We guess $k$ vertices $u_1,\ldots,u_k$ in $G$, and for each choice we check whether $G$ has an anchored topological minor $H'$ with the terminals $u_1,\ldots,u_k$ isomorphic to $H$, such that the isomorphism maps $v_i$ to $u_i$ for $i\in\{1,\ldots,k\}$. 
Then $H$ is an induced topological minor of $G$ if and only if $H'$ is an anchored topological minor of $G$ for some choice of the vertices $u_1,\ldots,u_k$. We check
the latter statement by applying the $n^{O(1)}$ time algorithm of Corollary~\ref{cor:anchor}. 
As we have $O(n^k)$ possible choices, we only have to do this $O(n^k)$ times.

If $H$ has one or more isolated vertices, we do the following. First, we guess a set of independent vertices in $G$ that correspond to the isolated vertices of $H$. 
If we cannot find such a set of isolated vertices in $G$, then we stop and return {\tt no}. 
Otherwise we remove the  guessed vertices from $G$ together with their neighbours. We check for each choice, whether the resulting graph~$G'$ contains the graph~$H'$ obtained from $H$ by removing its isolated vertices as an induced topological minor. 
Note that such a new instance of {\sc Induced Topological Minor} is trivial if $G'$ or $H'$ has no vertices.
Since we have  $O(n^k)$ new graphs $G'$, this only adds a $O(n^k)$ factor to the running time.
We conclude that the total running time is $n^{k+O(1)}$, as desired.
\qed
\end{proof}

If $H$ is a part of the input, or if we parameterize the problem by the size of $H$, then the problem becomes hard even for 
a subclass of AT-free graphs. 
A graph is \emph{cobipartite} if its vertex set can be partitioned into two cliques. Hence, the class of cobipartite graphs forms a subclass of the class of AT-free graphs.

\begin{figure}[ht]
\centering\scalebox{0.75}{\input{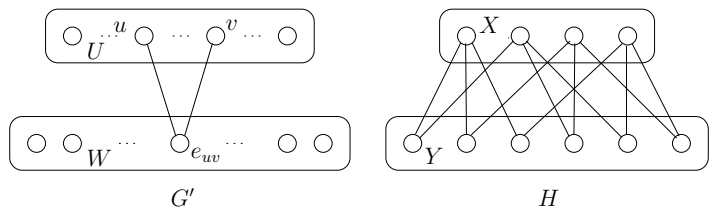}}
\caption{Construction of the graphs $G'$ and $H$ in the proof of Theorem~\ref{thm:W-hard} for $k=4$. 
\label{fig:hard}}
\end{figure}

\begin{theorem}\label{thm:W-hard}
The {\sc Induced Topological Minor} problem is \classNP-complete for cobipartite graphs, and
\classW{1}-hard for cobipartite graphs parameterized by $|V_H|$.
\end{theorem}

\begin{proof}
We prove \classW{1}-hardness.  This proof immediately implies \classNP-hardness, as we reduce from the well-known \classW{1}-complete {\sc Clique} problem~\cite{DowneyF99}. For a graph~$G$ and a parameter $k$, this problem is to test whether $G$ has a clique of size $k$. 

Let $(G,k)$ be an instance of {\sc Clique}. We assume without loss of generality  that $k\geq 5$ and that $G$ has minimum degree at least~4.
We construct a cobipartite graph~$G'$ as follows:
\begin{itemize}
\item[$\bullet$] create a copy of $V_G$ and construct a clique $U$ on these vertices;
\item[$\bullet$] for each edge $uv\in E_G$, create a vertex~$e_{uv}$ adjacent to $u$ and $v$;
\item[$\bullet$] construct the clique $W=\{e_{uv}\mid uv\in E_G\}$. 
\end{itemize}
Now we construct the graph~$H$:
\begin{itemize}
\item[$\bullet$] create a clique $X$ of size $k$ with vertices $x_1,\ldots,x_k$;
\item[$\bullet$] create a clique $Y$ with $\frac{1}{2}k(k-1)$ vertices $y_{ij}$, $1\leq i< j\leq k$;
\item[$\bullet$] for each pair $i,j$ with $1 \leq i < j \leq k$, make $x_i$ and $x_j$ adjacent to $y_{ij}$.
\end{itemize}
The construction is shown in Fig.~\ref{fig:hard}. 
We prove that $G$ has a clique of size $k$ if and only if $G'$ contains $H$ as an induced topological minor.

First suppose that $G$ has a clique $\{u_1,\ldots,u_k\}$. Then the subgraph of $G'$ induced by this clique and the set of vertices $\{e_{u_iu_j}\mid 1\leq i<j\leq k\}$ is isomorphic to $H$, i.e, $H$ is an induced subgraph of $G'$, and therefore an induced topological minor of $G'$.

Now suppose that $G'$ contains $H$ as an induced topological minor. Then $G'$ has an induced subgraph~$H'$ that is a subdivision of $H$. 
We call the vertices of $H'$ that are vertices of $H$ {\it branch vertices} and all other vertices of $H'$ {\it subdivision} vertices.

Let $x_i'$ for $1\leq i\leq k$ and $y_{ij}'$ for $1\leq i<j\leq k$ be the branch vertices of $H'$, where $x_i'$ and $y_{ij}'$ correspond to the vertices $x_i$ and $y_{ij}$ of the graph~$H$ respectively. Let $X'=\{x_i'\mid 1\leq i\leq k\}$ and 
$Y'=\{y_{ij}'\mid 1\leq i<j\leq k\}$.

Suppose that $H'$ has a subdivision vertex~$z$. Assume that $z\in U$. Since $d_{H'}(z)=2$, $H'$ contains at most two other vertices of $U$. Hence, $W$ contains at least $k-2\geq 3$ vertices of $X'$ and at least one vertex~$y_{ij}'\in Y'$. Then $y_{ij}'$ is adjacent to at least three vertices of $X'$, a contradiction. For the case $z\in W$ we use the same arguments, and
conclude that $H'$ has no subdivision vertices.

Assume that there are vertices $y_{ij}'\in U$ and $y_{pq}'\in W$. Since $k\geq 5$, $U$ or $W$ contains at least three vertices of $X'$. Note that $y_{ij}'$ is adjacent to at least three vertices of $X'$ in the first case and that $y_{pq}'$ has at least three such neighbours in the second, a contradiction. Therefore, either $Y'\subseteq U$ or $Y'\subseteq W$.

Now we prove that either $X'\subseteq U$ or $X'\subseteq W$. To obtain a contradiction, assume that there are $x_i'\in U$ and $x_j'\in W$. If $Y'\subseteq U$, then all vertices of $Y'$ are adjacent to $x_i'$, but since $k>2$, at least one vertex of $Y'$ is not adjacent to $x_i'$, a contradiction. We get the same conclusion if $Y'\subseteq W$. 

The sets $X'$ and $Y'$ clearly cannot be subsets of the same clique $U$ or $W$. 
Moreover, because each vertex of $W$ is adjacent to exactly two vertices of $U$ and each vertex of $X'$ is adjacent to at least four vertices of $Y'$, $X'\subseteq U$ and $Y'\subseteq W$. By the construction of $G'$ and $H$, this means that $X'$ is a clique of size $k$ in $G$.
This concludes the proof of Theorem~\ref{thm:W-hard}.\qed
\end{proof}

\section{Induced Paths and Trees through Specified Vertices}\label{s-path+tree}
In this section, we consider the $k$-{\sc in-a-Tree} and $k$-{\sc in-a-Path} problem. Recall that this is the problem of detecting an induced tree respectively path containing a set $S$ of $k$ specified vertices (called terminals as well).

We use the following structural property of AT-free graphs.
\begin{lemma}\label{lem:separ}
Let $G$ be an AT-free graph. For two vertices
 $u$ and $v$ of $G$ let $P$ be an induced $(u,v)$-path of length at least~$4$  in $G$. 
Let $G'$ be the subgraph obtained from $G$ after removing the vertices of $N[V_P\setminus\{u,v\}]$. 
If $G_1$ and $G_2$ are connected components of $G'$ containing a neighbour of $u$ and $v$ in $G$, respectively, then $G_1\neq G_2$.
\end{lemma}
\begin{proof}
For contradiction, suppose that $G_1=G_2$. Let $u'$ and $v'$ denote the neighbours of $u$ and $v$ in $G$ that belong to $G_1=G_2$.
As $G_1$ is connected, $G_1$ contains an induced $(u',v')$-path $P'$. 
As $G_1$ contains no vertices from $N[V_P\setminus\{u,v\}]$, we then find that 
 $P'$ contains no such vertices either.
As $P$ and $P'$ are induced and $P$ has length at least~4, this means that the union $C$ of $P$ and $P'$ is an induced cycle on at least six vertices. 
The latter implies that $C$ has three mutually non-adjacent vertices that form an asteroidal triple in $G$, a contradiction.\qed
\end{proof}

For $k$-{\sc in-a-Tree}, we need the following definitions and well-known observation. 
A tree $T$ is a \emph{caterpillar} if there are two leaves $x$ and $y$ such that all other vertices are either vertices of the $(x,y)$-path $P$ or are adjacent to the vertices of $P$. We say that $P$ is a \emph{central path} of $T$.
 
\begin{lemma}\label{lem:caterp}
If $T$ is an induced tree in an AT-free graph~$G$ then $T$ is a caterpillar.
\end{lemma}

We are now ready to prove the results of this section.

\begin{theorem}\label{thm:ktree}
The $k$-{\sc in-a-Tree} problem can be solved in polynomial time for AT-free graphs 
even if $k$ is part of the input.
\end{theorem}

\begin{proof}
Using Lemma~\ref{lem:caterp}, we look for an induced path $P$ between two terminals such that $V_P$ and $S$ induce a caterpillar with central path $P$, that is, all terminal vertices  that are not in $P$ are pairwise non-adjacent and each has exactly one neighbour in $P$.

We build this caterpillar step by step, using a dynamic programming algorithm that is somewhat similar to the algorithm used in the subroutine {\tt Component}$(i,X,Y)$ in Section~\ref{sec:comp}. We use similar notations as well. The idea will be to ``trace'' the central path through the graph. However, in contrast to the subroutine {\tt Component}$(i,X,Y)$, we have no ``template'' to go by and need to build the central path and the caterpillar step by step.

Before we show how the dynamic programming table is constructed and updated, we first give some additional terminology.
For an ordered set (a sequence of distinct elements) $R=\{r_1,\ldots,r_c\}$, we say that $r_1$ is the \emph{first} element of~$R$ and $r_c$ is the \emph{last} element. The first and last elements of~$R$ are denoted $\mathrm{first}(R)$ and $\mathrm{last}(R)$ respectively. We also say that an ordered set of vertices $R=\{r_1,\ldots,r_c\}$ of a graph~$G$ induces a path if $r_1\ldots r_c$ is an induced path in $G$. 
Let $P$ be an induced $(u,v)$-path $P$ in a graph~$G$ and let $c$ be a positive integer.
By $P(c)$ we denote the ordered set of the last $c$ vertices of~$P$ in the path order. If $P$ has length at most $c$, then $P(c)$ is the ordered set of the vertices of~$P$. Observe that $\mathrm{last}(P(c))=v$.

For an ordered set of vertices $R$ that induces a path with end-vertices $\mathrm{first}(R)$ and $\mathrm{last}(R)$, $A(R)$ denotes the set of terminals in the connected component of $G[(V_G\setminus N_G[V_R\setminus\{\mathrm{first}(R),\mathrm{last}(R)\}])\cup\{\mathrm{last}(R)\}]$ that contains $\mathrm{last}(R)$. We can think of $A(R)$ as the terminals that are ahead of $R$.

Our dynamic programming algorithm creates tables that store entries with the following fields:
\begin{itemize}
\item[$\bullet$] a non-negative integer $\ell$;
\item[$\bullet$] a terminal~$x\in S$ (the starting terminal of the caterpillar); 
\item[$\bullet$] a non-negative integer~$r$ (the number of terminals in the caterpillar as constructed so far);
\item[$\bullet$] a vertex~$z$ of $G$ (the last vertex of the central path of the caterpillar as constructed so far; note that $z$ might be but is not necessarily equal to $x$); and
\item[$\bullet$] an ordered set $R$ of at most five vertices of $G$ with $z=\mathrm{last}(R)$ that induces a path,
\end{itemize}
such that there is an induced $(x,z)$-path $P$ in $G$
and a partition $A,B,N$ of $S$ (some sets can be empty)
with the following properties:
\begin{itemize}
\item[$\bullet$] $P(5)=R$;
\item[$\bullet$] $N=S\cap (N_G[V_R\setminus\{\mathrm{first}(R),z\}])$;
\item[$\bullet$] $|B\cup N|=r$;
\item[$\bullet$] $G[V_P\cup B\cup N]$ is a caterpillar; and
\item[$\bullet$] for $A=A(R)\setminus\{z\}$, it holds that $A\cap V_P=\emptyset$ and vertices of $A$ are not adjacent to the vertices of $(V_P\cup N)\setminus\{z\}$.
\end{itemize}
We can think of $A$ as the set of terminals that are still ahead, of $B$ as the set of terminals that are behind, and of $N$ as the set of terminals that are currently in the neighbourhood of $R$. Also note that since $z \in R$, it would not be necessary to make $z$ a separate field; however, we will do so for notational convenience.

The tables are constructed consecutively for $\ell=0,\ldots,|V_G|-1$. We stop if one of the following occurs:
\begin{itemize}
\item[$\bullet$] the table contains a entry with $r=k$, and we return {\tt yes} in this case; or
\item[$\bullet$] we cannot construct any entry for the next value of~$\ell$, and we return {\tt no} in this case.
\end{itemize}

\medskip
\noindent
{\bf The table for \boldmath$\ell=0$.} This table only contains entries 
$(0,0,x,z,R)$, where $x\in S$, $z=x$, and $R=\{x\}$.

\medskip
\noindent
{\bf The tables for \boldmath$\ell>0$.} Assume that the table for $\ell-1$ has been constructed. We consider all vertices $v\in V_{G}$ for which the table for $\ell-1$ contains a entry $(\ell-1,x,r,w,R)$ such that
$v$ is adjacent to $w$ but is not adjacent to any other vertices of $R$. Then, for 
each $v$ and $(\ell-1,x,r,w,R)$, we create new entries in the table for $\ell$. In order to do this, we construct the sets $A=A(R)\setminus\{w\}$
and $N=S\cap (N_G[V_R\setminus\{\mathrm{first}(R),\mathrm{last}(R)\}])$. 
Then we construct $R'$ by adding $v$ to $R$ as the last element. 
Let $U=(N_G(w)\cap A)\setminus\{v\}$. Then we check whether the following holds:
\begin{itemize}
\item[$\bullet$] each vertex of $U$ is adjacent to exactly one vertex of $R'$ (i.e.\ only to the vertex~$w$);
\item[$\bullet$] each vertex of $U$ is not adjacent to the vertices of $N\setminus\{w\}$;
\item[$\bullet$] for the set $A'=A\setminus(\{v\}\cup U)$, it holds that $A'=A(R')\setminus\{v\}$.
\end{itemize}
The first two verify that we are indeed inducing a caterpillar, whereas the last verifies that we are expanding the caterpillar in the correct direction (towards the remaining terminals that are ahead).
If these conditions are fulfilled, then we set $R''=R'\setminus \{\mathrm{first}(R)\}$ if $|R|=5$, and 
$R''=R'$ otherwise. Let $r'=r+|A|-|A'|$. Then we include the entry $(\ell,x,r',v,R'')$ in the table.  

Observe that by Lemma~\ref{lem:separ}, any vertex of $G$ adjacent to $z=\mathrm{last}(R)$ but non-adjacent to other vertices of $R$, is not adjacent to the vertices of the $(x,z)$-path $P$, except the last vertex~$z$, and is not adjacent to the vertices of $B$. Then the correctness of our algorithm follows from the description and this observation.

It remains to prove that the algorithm is polynomial. 
Let $n=|V_G|$. Since $r\leq k$, there are $k$ possibilities for the terminal~$x$ and at most $n^5$ ways to choose $z$ and $R$. Hence, the total size of the table for any $\ell$ is at most $k^2\cdot n^5$. Since $\ell\leq n-1$, the running time is polynomial.
\end{proof}

\begin{theorem}\label{thm:kpath}
The $k$-{\sc in-a-Path} problem can be solved in polynomial time for AT-free graphs 
even if $k$ is part of the input.
\end{theorem}
\begin{proof}
We use a dynamic programming algorithm that is constructed along the same lines as for the {\sc $k$-in-a-Tree} problem.

We have to find an induced path $P$ between two terminals such that $S\subseteq V_P$.
As in the proof of Theorem~\ref{thm:ktree}, for an ordered set of vertices $R$ that induces a path with the end-vertices $\mathrm{first}(R)$ and $\mathrm{last}(R)$, $A(R)$ is the set of terminals in the connected component $G[(V_G\setminus N_G[V_R\setminus\{\mathrm{first}(R),\mathrm{last}(R)\}])\cup\{\mathrm{last}(R)\}]$ that contains $\mathrm{last}(R)$.

Our dynamic programming algorithm creates tables that store entries with the following fields:
\begin{itemize}
\item[$\bullet$] a non-negative integer $\ell$;
\item[$\bullet$] a terminal~$x\in S$ (the starting terminal of the path); 
\item[$\bullet$] a non-negative integer~$r$ (the number of terminals on the path as constructed so far);
\item[$\bullet$] a vertex~$z$ of $G$ (the last vertex of the path as constructed so far; note that $z$ might be but is not necessarily equal to $x$); and
\item[$\bullet$] an ordered set $R$ of at most five vertices of $G$ with $z=last(R)$ that induces a path,
\end{itemize}
such that there is an induced $(x,z)$-path $P$ in $G$,
and a partition $A,B,N$ of $S$ (some sets can be empty)
with the following properties:
\begin{itemize}
\item[$\bullet$] $P(5)=R$;
\item[$\bullet$] $N=S\cap R$;
\item[$\bullet$] $B\subseteq V_P\setminus R$;
\item[$\bullet$] $|B\cup N|=r$; and
\item[$\bullet$] for $A=A(R)\setminus\{z\}$, it holds that $A\cap V_P=\emptyset$ and vertices of $A$ are not adjacent to the vertices of $V_P\setminus\{z\}$.
\end{itemize}
The intuition for $A$, $B$, and $N$ are similar as before.

The tables are constructed consecutively for $\ell=0,\ldots,|V_G|-1$. We stop if one of the following occurs:
\begin{itemize}
\item[$\bullet$] the table contains a entry with $r=k$, and we return {\tt yes} in this case;
\item[$\bullet$] we cannot construct any entry for the next value of $\ell$, and we return {\tt no} in this case.
\end{itemize}

The tables are constructed and updated as follows.

\medskip
\noindent
{\bf The table for \boldmath$\ell=0$.} This table only contains entries 
$(0,0,u,z,R)$, where $u\in S$, $z=u$, and $R=\{u\}$.

\medskip
\noindent
{\bf The tables for \boldmath$\ell>0$.} Assume that the table for $\ell-1$ has been constructed. We consider all vertices $v\in V_{G}$ for which the table for $\ell-1$ contains a entry $(\ell-1,x,r,w,R)$ such that
$v$ is adjacent to $w$ but is not adjacent to any other vertices of $R$. Then, for 
each $v$ and $(\ell-1,x,r,w,R)$, we create new entries for the table for $\ell$. In order to do this, we construct the set $A=A(R)\setminus\{w\}$.
Then we construct $R'$ by adding $v$ to $R$ as the last element.
Let $A'=A\setminus\{v\}$. 
Then we check whether $A'=A(R')\setminus\{v\}$. This verifies that we are expanding the path in the correct direction (towards the remaining terminals that are ahead).
If this holds, then we set $R''=R'\setminus \{\mathrm{first}(R)\}$ if $|R|=5$, and 
$R''=R'$ otherwise. Let $r'=r+|A|-|A'|$. 
Then we include the entry $(\ell,x,r',v,R'')$ in the table.  

Observe that by Lemma~\ref{lem:separ}, any vertex of $G$ adjacent to $z=\mathrm{last}(R)$ but non-adjacent to other vertices of $R$, is not adjacent to the vertices of the $(x,z)$-path $P$, except the last vertex~$z$. Then the correctness of our algorithm follows from the description and the preceding observation.

It remains to prove that the algorithm is polynomial. 
Let $n=|V_G|$. Since $r\leq k$, there are $k$ possibilities for the terminal~$x$ and at most $n^5$ ways to choose $z$ and $R$. Hence, the total size of the table for any $\ell$ is at most $k^2\cdot n^5$. Since $\ell\leq n-1$, the running time is polynomial.
\qed
\end{proof}

\section{Concluding Remarks}\label{s-con}
We have presented a polynomial-time algorithm that solves {\sc Induced Disjoint Paths} for AT-free graphs, and we used this algorithm for a polynomial-time algorithm that solves $H$-{\sc Induced Topological Minor} on this graph class for every fixed graph~$H$. We complemented the latter result by proving that  {\sc Induced Topological Minor}, restricted to cobipartite graphs, is \classNP-complete and \classW{1}-hard when parameterized by~$|V_H|$.
We also showed that the problems $k$-{\sc in-a-Tree} and $k$-{\sc in-a-Path} can be solved in polynomial time for AT-free graphs even if $k$ is part of the input.

\subsection{Open Problem 1}

Motivated by our application on testing for induced topological minors, we assumed that all terminal pairs in an instance $(G,S)$ of {\sc Induced Disjoint Paths} are distinct. 
For general graphs, we can easily drop this assumption by replacing a vertex~$u$ representing $\ell\geq 2$ terminals by $\ell$ new mutually non-adjacent vertices, each connected to all neighbours of $u$ via subdivided edges.
This yields an equivalent instance of {\sc Induced Disjoint Paths}, in which all terminal pairs are distinct. 
However, we cannot apply this reduction for AT-free graphs, because it may create asteroidal triples. 
Hence, the complexity of this more general problem remains an open question.

We note that the special case in which all terminal pairs coincide, that is, in which $(s_1,t_1)=\cdots = (s_k,t_k)$ is already
 \classNP-compete for $k=2$~\cite{Bi91,Fe89} and solvable in $O(n^2)$ time on $n$-vertex planar graphs 
 for arbitrary $k$~\cite{MRSS1994}.
We can solve this case in polynomial time for AT-free graphs as follows.

Let $G$ be an AT-free graph with $k$ terminal pairs that are copies of the same terminal pair $(s,t)$. 
If $k=1$, then the problem is trivial: we have to find an $(s,t)$-path in $G$. Hence, we assume that $k\geq 2$. 
Suppose that $(s,t)$-paths $P_{1},\ldots,P_{k}$ form a solution.
Then each path $P_i$ has length at most~4. Otherwise, if there is a path $P_i$ of length at least~5, this path and any other path $P_j$ would induce a cycle with at least six vertices, but an AT-free graph has no induced cycles of length at least~6. Moreover, at most two paths $P_i,P_j$ have length at least~3. To obtain a contradiction, assume that three paths $P_i,P_j,P_r$ have length at least~3. Let $u,v,w$ be the vertices adjacent to $t$ in $P_i,P_j,P_r$ respectively. Then $u,v,w$ form an asteroidal triple, a contradiction. 
By these observations, we consider cases when a solution has $r=0,1,2$ paths of length at least~3 and at most~4. For each $r$, we guess these paths by brute force. Then we 
construct the graph~$G'$ by the deletion of all internal vertices of these paths, together with their neighbours different from $s,t$. Now we have to find $k'$ $(s,t)$-paths in $G'$ of length~2, where $k'=k-r$ if $s,t$ are non-adjacent, and $k'=k-r-1$ if $s,t$ are neighbours. These paths exist if and only if the graph~$H=G'[N_{G'}(s)\cap N_{G'}(s)]$ has an independent set of size at least $k'$. It remains to observe that $H$ is an AT-free graph, and that the {\sc Independent Set} problem can be solved in polynomial 
time for AT-free graphs~\cite{BroersmaKKM99}.

\subsection{Open Problem 2}
Our three algorithms for solving {\sc Induced Disjoint Paths}, $k$-{\sc in-a-Tree} and $k$-{\sc in-a-Path} on AT-free graphs are polynomial-time algorithms, but the polynomials in the running time bounds have high degree. It  may be interesting to construct more efficient algorithms even for subclasses of the AT-free graphs, such as the class of permutation graphs or the class of cocomparability graphs.
In contrast to the class of AT-free graphs, both these graph classes are known to have a geometric intersection model.  

\subsection{Open Problem 3}
Our last  open problem is whether it is possible to extend our result for {\sc Induced Disjoint Paths} to some superclass of AT-free graphs, such as graphs of bounded asteroidal number. An {\it asteroidal set} in a graph~$G$ is an independent set $S\subseteq V_G$, such that every triple of vertices of $S$ is asteroidal. The {\it asteroidal number}
introduced by Walter~\cite{Walter78} (see also Kloks, Kratsch and M\"uller~\cite{KKM01})  is the size of a largest asteroidal set in $G$.
Note that complete graphs are exactly those graphs that have asteroidal number at most~$1$, and that AT-free graphs are exactly those graphs that have asteroidal number at most~2. 
There exist problems for which polynomial-time algorithms for AT-free graphs could be extended to polynomial-time algorithms
for graphs of bounded asteroidal number 
(see, e.g., \cite{BroersmaKKM99,KM12,KMT03}).
Is this also possible for {\sc Induced Disjoint Paths}, or will this problem become \classNP-complete for some constant value of the asteroidal number?

\end{document}